\documentclass[journal,comsoc]{IEEEtran}

\usepackage[T1]{fontenc} 

\usepackage{cite} 

\usepackage{graphicx} 

\usepackage{amsmath} 

\usepackage{amssymb}

\usepackage{amsthm}

\usepackage[cmintegrals]{newtxmath} 

\usepackage{algorithmicx}

\usepackage{algpseudocode}

\usepackage{algorithm}

\usepackage{float}

\usepackage{glossaries}

\usepackage{xcolor}

\usepackage{subcaption}

\usepackage{url}  

\usepackage{multirow}

\usepackage{setspace}

\usepackage{balance}

\newtheorem{theorem}{Theorem}

\theoremstyle{remark}

\renewcommand{\figurename}{Figure} 

\begin{document}

\title{A Reliability-Aware, Delay Guaranteed, and Resource Efficient Placement of Service Function Chains in Softwarized 5G Networks}

\author{Prabhu Kaliyammal Thiruvasagam, Vijeth J Kotagi, and C Siva Ram Murthy, \textit{Fellow, IEEE} \\ 
	Indian Institute of Technology Madras, Chennai 600036, India \\
	prabhut@cse.iitm.ac.in, vijethjk@cse.iitm.ac.in, murthy@iitm.ac.in}
 
\maketitle

\begin{abstract}
Network Functions Virtualization (NFV) allows flexibility, scalability, agility, and easy manageability of networks by leveraging the features of virtualization and cloud computing technologies. However, softwarization of network functions imposes many challenges. Reliability and latency are major challenges in NFV-enabled 5G networks that can lead to customer dissatisfaction and revenue loss. In general, redundancy is used to improve the reliability of communication services. However, redundancy requires the same amount of additional resources and thus increases cost. In this work, we address the reliability-aware, delay guaranteed, and resource efficient  Service Function Chain (SFC) placement problem in softwarized 5G networks. First, we propose a novel SFC subchaining method to enhance the reliability of an SFC without backups. If reliability requirement is not met after subchaining method, we add backups to VNFs to meet the reliability requirement. Then, we formulate the reliable SFC placement problem as an Integer Linear Programming (ILP) problem in order to solve it optimally. Owing to high computational complexity of the ILP problem for solving large input instances, we propose a modified stable matching algorithm to provide near-optimal solution in polynomial time. By extensive simulations we show that our proposed solutions consume lesser physical resources compared to state-of-the-art solutions for provisioning reliable communication services.
\end{abstract}


\section{Introduction}
\IEEEPARstart{T}{he} current traditional networks have certain limitations such as long design and testing time to bring network functions to market, location dependent functionalities (migration of network functions is not an easy task without interrupting services), and not being able to support the required flexibility and scalability to meet the demands of communication services. Network operators need to design and deploy additional network functions to adapt to new technology and accommodate the growth of mobile connected devices, which results in increased capital expenditure (CAPEX) and operational expenditure (OPEX). To reduce the costs, and make networks more flexible and scalable to handle the ever increasing demands and future business opportunities, communication service providers and network operators are moving towards softwarized 5G networks. 

Softwarization in 5G networks to support services such as enhanced mobile broadband and ultra-reliable low-latency communications has revolutionized the networking industry. It is expected that 5G networks will meet the stringent requirements of communication services of various industry verticals and business models of 2020 and beyond \cite{5g_req}. Network Functions Virtualization (NFV) and Software-Defined Networking (SDN) are the two new promising technologies in the field of networking, which are designed to overcome the limitations of traditional networks. NFV and SDN are considered as key technology enablers for softwarization in 5G networks \cite{zarrar}.

NFV allows network functions (or middleboxes) to run as software modules on commercial-off-the-shelf servers rather than on specialized hardware appliances. Such virtualized software modules are called as Virtual Network Functions (VNFs). SDN \cite{Kreutz_2015} enables network programmability by decoupling control plane and data plane of networking devices. Data plane and control plane separation makes networking devices controllable through a centralized entity. The centralized control plane or controller can orchestrate multiple data plane flows dynamically. SDN can be used to route the traffic dynamically through multiple VNFs. NFV leverages virtualization, cloud computing, and SDN technologies to provide anything as a service (e.g., core network as a service, security as a service, etc.) dynamically over the network, and reduces CAPEX and OPEX. NFV provides an effective way to design, deploy, and manage network functions and services over the cloud environment. 

Traditionally, network/communication services are provided through one or more network functions to deliver an end-to-end (e2e) service. SFC involves instantiation of an ordered list of network/service functions (e.g., firewalls, load balancers, and mobile network gateways), and connecting them together as a chain of network functions to provide e2e services \cite{SFC_PS} \cite{Haeffner-SFC}. SFC is used to design a tailor-made specific communication service based on the demand. Service chain deployment includes SFC design and placement of SFCs. NFV facilitates easy provisioning of services by dynamically placing VNFs in the virtual environment and chaining them together as an SFC. Efficient mapping and placement of SFCs onto substrate nodes is a challenging task \cite{SFC_PS} \cite{Medhat_2017}.   

Although NFV and SDN provide many benefits in terms of cost reduction and flexible management of resources to 
dynamically provide diverse services, they create avenues for reliability, availability, and latency related issues. Particularly, softwarization of network and service functions impose higher possibility of network failures due to software issues than due to hardware issues. For instance, failure of a VNF or a virtual link in an SFC will bring down the entire chain and disrupt the service which may result in customer dissatisfaction and revenue loss. Failures may happen both at the substrate network and the virtual network, but the frequency of failures is higher at virtual networks than at substrate networks \cite{Benz}. In NFV infrastructure, a VNF may fail to do its intended function due to various reasons such as software bugs, misconfiguration, API failures, malicious attacks, and network operator errors. Abrupt failure of any function or connected link results in delay, data loss, and resource wastage. Since failure of a single component of an SFC has high impact on the ongoing services and revenue, reliable service provisioning is of paramount importance in NFV-enabled 5G networks. Therefore, merely placing primary VNFs of an SFC is not sufficient to ensure the service continuity in case of failures. Service continuity is not only expected from user side, but it is also considered as a regulatory requirement for critical infrastructures like telecom networks \cite{NFV_REL_001}. Hence, reliability is an important issue when purpose-built dedicated hardware based network function with high reliability is moved to off-the-self general purpose servers.

The reliability of a communication service is estimated based on the reliability of constituent network functional blocks \cite{NFV_REL_003}. Higher reliability of SFC minimizes the impact of service outages due to unexpected VNF failures. Another important aspect of NFV-enabled 5G network is meeting Service Level Agreements (SLAs) in terms of delay, availability, and reliability. A common approach for achieving higher reliability and meeting delay constraints is placing redundant network elements (also called as backups) \cite{Fan_2015} \cite{Qu1_2018} \cite{Sun}. However, such an approach is expensive and ineffective in terms of utilization of available resources. In this paper, we propose novel methods to address the reliability issues and efficiently place SFCs onto the substrate network.
 
The significant contributions of this paper are listed below. 
\begin{itemize}
	\item We present two ways of assigning backups to an SFC and propose a novel subchaing method to enhance the reliability of the SFC without redundancy. 
	
	\item We propose an algorithm to calculate the reliability of multiple subchains of an SFC. In the case that the reliability enhancement obtained by SFC subchaining method is not sufficient, we propose an algorithm to guarantee the reliability requirement of diverse service requests with minimal redundant resources. 
	
	\item We formulate the reliable SFC placement problem as an Integer Linear Programming (ILP) problem, and prove that the SFC placement 
	problem is NP-hard. 
	
	\item We use JuMP and Gurobi optimization solver for modeling and solving the ILP problem, respectively. To overcome high computational complexity of ILP for large input instances, we propose a modified matching algorithm to obtain near-optimal solution in polynomial time. 
	
	\item Finally, we evaluate the performance of our proposed solutions via simulation, and show that our proposed solutions outperform the state-of-the-art related works. 	
	
\end{itemize}
The rest of this paper is organized as follows. Section II presents background and related work. Section III presents system model and problem definition. Section IV describes a novel SFC subchaining method to enhance the reliability, and algorithms to guarantee the reliability requirement of SFCs. Section V presents the ILP formulation and solution for efficient reliable SFC design placement. We evaluate and analyze the performance of the proposed solutions in section VI. Finally, we conclude the paper with future work in section VII.

\section{Background and Related Work}
NFV concept was introduced in 2012 by a group of leading telecom operators, who later created the NFV Industrial Specification Group (NFV ISG) in European Telecommunications Standards Institute (ETSI) to call for collaboration and standardization activities \cite{NFV_WP1_2012}. ETSI NFV ISG has been creating a number of technical requirement documents including NFV reference architectural frameworks \cite{NFV_002_2014} \cite{NFV_MANO}. NFV architectural framework \cite{NFV_002_2014} consists of NFV Infrastructure (NFVI) layer, VNF layer, service layer, and management and orchestration (MANO) layer. Both physical and virtual resources reside at the NFVI layer, which form a cloud network infrastructure to provide services through VNF layer and SFC layer. MANO layer \cite{NFV_MANO} is responsible for dynamic resource management and life-cycle management of VNFs and services.

Since NFV is considered as one of the key technology enablers for 5G softwarized networks \cite{zarrar}, it has gained significant attention of the researchers from both industry and academia. VNF/SFC placement strategy is not standardized by standard development organizations and it is up to the choice of network operators and service providers to use their own strategies to efficiently allocate the resources for service provisioning. In the context of VNF/SFC placement in NFV-enabled networks, \cite{Bari_2016} \cite{Taleb2_2019} and their corresponding references mainly focused on reducing operational cost and network resource allocation with the goal of maximizing the revenue of network operators and service providers. However, simply placing VNFs in order to reduce the resource consumption or cost may not meet the SLA requirements of service requests. 
\begin{table*}[]
	\centering
	\caption{Comparison of existing related studies}
	\label{tab:relatedWork_comparison}
	\begin{tabular}{|c|c|c|c|c|c|c|}
		\hline
		References & \begin{tabular}[c]{@{}c@{}}Latency-aware\\ SFC placement\end{tabular} & \begin{tabular}[c]{@{}c@{}}Reliability-aware\\ SFC placement\end{tabular} & \begin{tabular}[c]{@{}c@{}}Full backup\\ is required\end{tabular} & \begin{tabular}[c]{@{}c@{}}Partial backup\\ is enough\end{tabular} & \begin{tabular}[c]{@{}c@{}}Physical node\\ reliability\end{tabular} & \begin{tabular}[c]{@{}c@{}}VNF\\ reliability\end{tabular} \\ \hline
		\cite{Alleg, Alameddine1, Alameddine2, Garg_2019, Bi_2019, Harut_2019} & $\checkmark$ & $\times$ & $\times$ & $\times$ & $\times$ & $\times$ \\ \hline 			
		\cite{Fan_2015} & $\times$ & $\checkmark$ & $\checkmark$ & $\times$ & $\checkmark$ & $\times$ \\ \hline 
		\cite{Qu1_2018} & $\times$ & $\checkmark$ & $\checkmark$ & $\times$ & $\times$ & $\checkmark$ \\ \hline
		\cite{Herker_2015} & $\times$ & $\checkmark$ & $\checkmark$ & $\times$ & $\checkmark$ & $\times$ \\ \hline  
		\cite{Ali_2016} & $\checkmark$ & $\checkmark$ & $\checkmark$ & $\times$ & $\checkmark$ & $\times$ \\ \hline 
		\cite{Ye_2016} & $\times$ & $\checkmark$ & $\checkmark$ & $\times$ & $\checkmark$ & $\times$ \\ \hline
		\cite{Taleb1_2016} & $\times$ & $\checkmark$ & $\times$ & $\times$ & $\checkmark$ & $\times$ \\ \hline 
		\cite{Kanizo_2017} & $\times$ & $\checkmark$ & $\checkmark$ & $\times$ &  $\times$ & $\checkmark$ \\ \hline
		\cite{Chantre1_2018} & $\times$ & $\checkmark$ & $\checkmark$ & $\times$ & $\checkmark$ & $\times$ \\ \hline 		 
		\cite{Ding_2017} & $\times$ & $\checkmark$ & $\checkmark$ & $\times$ & $\checkmark$ & $\checkmark$ \\ \hline 
		\cite{Engel_2018} & $\times$ & $\checkmark$ & $\checkmark$ & $\times$ & $\checkmark$ & $\times$ \\ \hline 
		Our work & $\checkmark$ & $\checkmark$ & $\times$ & $\checkmark$ & $\checkmark$ & $\checkmark$ \\ \hline  
		\end{tabular}	
\end{table*}

A set of works in the literature dealt with delay/latency aspects of VNF placement in addition to minimizing provisioning cost or resource consumption. A~Alleg \textit{et al} \cite{Alleg} modeled the delay-aware VNF placement and chaining problem as mixed integer quadratically constrained program to provide a solution with the goal of minimizing the resource consumption while meeting the SLA latency requirement. H~A~Alameddine \textit{et al} \cite{Alameddine1} \cite{Alameddine2} formulated SFC mapping, routing, and scheduling problem as mixed ILP problem to meet the deadline of service requests, and proposed heuristic algorithm for scalable networks. G~Garg \textit{et al} \cite{Garg_2019} proposed delay-aware VNF selection algorithm by considering the relationship between VNF delay and CPU utilization to increase the acceptance rate and throughput of service requests. Y~Bi \textit{et al} \cite{Bi_2019} proposed resource allocation method for ultra-low latency virtual network services in hierarchical 5G network. The problem was modeled as mixed ILP to obtain optimal solution, and data rate-based heuristic algorithm was proposed for scalable networks. D~Harutyunyan \textit{et al} \cite{Harut_2019} studied latency-aware SFC placement in 5G mobile networks by formulating the problem as ILP to minimize the e2e latency and service cost and proposed a heuristic algorithm to address the scalability issue. However, these research works mainly focused on latency-aware service provisioning and assumed that the network functions were active all the time without any failure. It may not be the case in real-time service provisioning scenario because there is a possibility for service interruptions due to service degradation and failure of network functions and resources. 

A few works considered reliability aspects of VNFs/services along with the objective of minimizing the resource consumption or operational cost. J~Fan \textit{et al} \cite{Fan_2015} proposed joint protection online algorithm which allocates a joint backup for two VNFs on a single server to enhance the reliability and reduce the physical resource consumption. However, in this work only physical node reliability was considered and two backup VNFs were assigned at each iteration. S~Herker \textit{et al} \cite{Herker_2015} proposed heuristic algorithm with two backup deployment strategies (with and without load balancer) for reliable VNF chains in the data center networks. However, the proposed strategies consumed more physical resources to meet the reliability requirements.  A~Hmaity \textit{et al} \cite{Ali_2016} proposed a method for VNF placement and resilient service chain provisioning, and formulated the problem as ILP problem to minimize the number of active physical nodes used while satisfying the latency constraints. However, the proposed method can solve only a small number of input instances due to the high computational complexity and consumed more than twice the amount of network resources to provide resiliency. Z~Ye \textit{et al} \cite{Ye_2016} proposed joint optimization of topology design and mapping of SFCs to minimize the total bandwidth consumption and compared the amount of consumption of resources in order to enhance the reliability of SFC requests. However, they did not provide any methods to guarantee the reliability and latency requirements of service requests. T~Taleb \textit{et al} \cite{Taleb1_2016} proposed restoration mechanism for VNF failures, particularly considered Mobility Management Entity (MME) control plane VNF failure restoration process by proposing bulk signalling and profile creation to reduce the load. However, the failure restoration mechanism focused on specific VNF type and hence it may not meet the reliability requirements of mission critical service requests in general. Y~Kanizo \textit{et al} \cite{Kanizo_2017} proposed a novel approach for planning and deploying backups optimally to guarantee the survivability of service chain. However, full backup was required to guarantee the survivability, and latency aspect was not considered. H~Chantre and N~Fonseca \cite{Chantre1_2018} proposed two redundancy based models for reliable broadcasting in 5G NFV-based networks, which determine the number of redundant VNFs required to meet the service reliability requirement. However, these methods required full backups, and latency aspect was not considered. W~Ding \cite{Ding_2017} \textit{et al} proposed cost-efficient redundancy scheme for enhancing the reliability of services in NFV-enabled networks. However, full backup was required for enhancing the reliability of services, and latency aspect was not considered. 

In Table \ref{tab:relatedWork_comparison}, we compare our approach with different reliability-aware and latency-aware VNF/SFC placement techniques proposed in the literature. The closest research works to this current work are \cite{Qu1_2018} \cite{Engel_2018} \cite{Pham_2018}. L~Qu \textit{et al} \cite{Qu1_2018} proposed reliability-aware approach for service chaining in carrier-grade softwarized networks. However, this approach considered only physical node reliability and required full backups to guarantee the reliability requirement, and did not take into account the latency aspect. A~Engelmann and A~Jukan \cite{Engel_2018} proposed parallelized VNF chaining to enhance the reliability of SFC, in which a large flow was split into multiple smaller sub-flows and each sub-flow was processed by replicated VNF in parallel. Though this approach aimed to process the sub-flows in parallel, VNFs were replicated to process subflows and dedicated full backup was assigned to enhance the reliability of network functions. Also, this work did not take into account the latency and resource minimization aspects. C~Pham \textit{et al} \cite{Pham_2018} proposed VNF placement for service chaining using sampling and matching approach. However, their work primarily focused on minimizing cost by considering energy consumption and traffic flow, and did not consider reliability and latency aspects.     
 
In our previous work \cite{PKT_2019}, we showed that dividing an SFC (single chain) into multiple subchains of lesser capacity enhances the reliability of the SFC. Also, we assumed that the underlying physical infrastructure is completely reliable. In the current work, we relax this assumption and consider that underlying substrate nodes are also subject to failures.  

Almost all the proposed methods in the literature used redundancy techniques directly to improve the reliability of communication services. However, such techniques are expensive and ineffective in terms of effective utilization of available resources. In this work, we propose a novel method to address the reliability requirement with minimal redundant resources and efficiently place SFCs on the substrate network. 

\section{Network Model and Problem Definition}
In this section, we present a network model of 5G architectural framework, which consists of underlying substrate network, virtual network, and service functions. 

\subsection{Substrate Network}
The physical/substrate network is modeled as an undirected graph denoted by $G_p = (\mathcal{N},\mathcal{L})$, where $\mathcal{N}$ represents a set of physical/substrate nodes which include both the processing and forwarding nodes and $\mathcal{L}$ represents the set of physical links. VNFs are placed on the processing nodes to process incoming service traffic and the forwarding nodes are used to interconnect a set of processing nodes and forward service traffic in the network. Each processing node $n \in \mathcal{N}$ has a finite resource capacity $c_n \in \mathbb{R}^+$ (CPUs, Memory, Disk Space, etc.)  and multiple VNFs can be hosted on a single processing node. Similarly, each substrate link has a finite bandwidth capacity and interconnects the physical nodes. Physical resources are virtualized to create virtual networks and controlled with the help of MANO and SDN controller.     

\subsection{Virtual Network}
We represent the virtual network as an undirected graph denoted by $G_v = (\mathcal{V},\mathcal{E})$, where $\mathcal{V}$ represents a set of VNFs (e.g., load balancer, firewall, intrusion detection system, proxy, mobility management entity, serving/packet gateway, home subscriber server, etc.) and $\mathcal{E}$ represents a set of virtual links which interconnect VNFs in the virtual network. Each VNF $v \in \mathcal{V}$ has a resource demand of $c_v$. Here, the resource demand for VNF $c_v$ represents the number of vCPUs required to process the incoming service traffic. Different types of VNFs can be customized and chained in various ways depending on the service type and requirements to obtain a set of templates. Each virtual link has certain bandwidth requirement (with respect to communication requirements  between the nodes) and interconnects two VNFs in a chain. VNFs are hosted on the physical servers and virtual links are created to interconnect the VNFs and to carry the network traffic over the physical links. The physical and virtual network resources together form cloud network infrastructure. 

\subsection{Service Function Chaining}
SFCs consist of logically interconnected multiple independent network/service functions. It is assumed that Communication Service Providers (CSPs) offer finite number of services using SFCs \cite{5g_req}. Let the set of all SFCs provided by a CSP be denoted by $\mathcal{S}$. Each SFC $s \in \mathcal{S}$ provides a particular service and is represented as a directed sequence of linear chain of functions, a special form of an acyclic directed graph, $G_s = (\mathcal{V}_s, \mathcal{E}_s)$, where $\mathcal{V}_s$ and $\mathcal{E}_s$ represent the set of VNFs in sequential order (i.e., it is the topological ordering such that the incoming traffic is processed by VNF i before it is being processed by VNF i+1 in the SFC chain) and the set of links that interconnect these VNFs, respectively. For example, consider a Voice over IP service request $s$, where the set of VNFs $\mathcal{V}_s$ required to cater to the service $s$ in an order are network address translation, firewall, and traffic monitor. Each VNF $v \in \mathcal{V}$ can be associated with only one SFC $s \in \mathcal{S}$ in order to avoid multiple SFC failures (due to VNF sharing). VNFs of an SFC can be placed on the same node in order to minimize the resources of active physical nodes used, bandwidth consumption, and inter VNF communication delay. Tens to hundreds of SFCs can be created to provide diverse services. As each SFC provides a particular service, we use terms SFC and service request interchangeably. Service orchestrator takes care of life-cycle-management of SFCs. 

\subsection{Service Requests and SLAs}
Different industry verticals have different service requirements based on the application/service type. SLAs between CSPs and customers define the specific requirements and contracts in terms of expected quality of service. CSPs should satisfy SLAs of service requests in order to gain high revenue. For instance, a service request $s \in \mathcal{S}$ can have specific requirements in terms of high data rate, bandwidth demand, maximum allowed delay, reliability, and availability. There can be a penalty policy for violation of SLAs. The penalty amount is paid to the users if a service requirement is not satisfied. 

\subsection{Problem Definition}
As an SFC placement, merely mapping primary VNFs of the SFC onto the substrate network is not sufficient for provisioning reliable communication services. In this work, we first redesign SFC requests to meet the reliability requirement and then efficiently place them onto the substrate network. We define the reliable SFC placement problem as two sub-problems. \\ 
\textit{\textbf{Sub-problem 1: Reliable SFC design.} Given a set of SFC requests, each with a specific delay and reliability requirements, design reliability-aware SFC service graphs such that it minimizes the redundant resources needed to guarantee the reliability requirements while satisfying the delay requirements.} \\ 	   
\textit{\textbf{Sub-problem 2: Placement of reliable SFC graphs.} Given a physical network graph and a set of reliability-aware SFC service graphs, find an efficient way of placing reliability-aware SFC service graphs onto the substrate network such that it minimizes the number of physical nodes required to provide reliable communication services.}   

\section{Reliable SFC Design}
\label{SFC_design}

\subsection{Enhancing Reliability of SFCs with Backups}
Usually, redundant backup VNFs are placed to enhance the reliability of the service chain. Backups are placed to ensure the service continuity in the case of VNF failures. We present two backup methods to enhance the reliability of an SFC. In this work, it is assumed that the incoming traffic of service request $s$ follows Poisson distribution with arrival rate $\lambda_s$ and the serving time of each VNF in an SFC follows an exponential distribution with the serving rate $\mu_s$ \cite{Gouareb}. Let $\psi_v$ be the mean response time of VNF $v \in \mathcal{V}$. Table \ref{tab:1} gives the list of notations and symbols used in this work.
\begin{table}[]
	\centering 
	\footnotesize
	\caption{List of Notations}
	\label{tab:1}
	\begin{tabular}{|l|l|}
		\hline
		$G_p = (\mathcal{N},\mathcal{L})$	&  Physical network $G_p$ with $\mathcal{N}$ nodes and $\mathcal{L}$ links\\ \hline
		$c_n \in \mathbb{R}^+$	&  Resource capacity of physical node $n \in \mathcal{N}$\\ \hline
		$G_v = (\mathcal{V},\mathcal{E})$	&  Virtual network $G_v$ with $\mathcal{V}$ VNFs and $\mathcal{E}$ virtual links\\ \hline
		$c_v$ & Resource requirement of vCPUs for a VNF $v \in \mathcal{V}$ \\ \hline
		$\mathcal{S}$ & Set of SFCs  \\ \hline
		$G_s = (\mathcal{V}_s,\mathcal{E}_s)$ & SFC $s \in \mathcal{S}$ with an ordered VNFs $\mathcal{V}_s$ and links $\mathcal{E}_s$ \\ \hline 	
		$c_s$ & Resource requirement of an SFC $s \in \mathcal{S}$ \\ \hline 	 
		$\mathcal{N}_s$ & SFC $s$ is placed in a physical node $n \in \mathcal{N}_s$  \\ \hline
		$p_n$ & Node $n$ is reliable with probability $p_n$ \\ \hline
		$p_v$ & VNF $v$ is reliable with probability $p_v$ \\ \hline 
		$r_s$ & Reliability of an SFC $s$ \\ \hline
		$\Delta_s$ & Reliability requirement of an SFC $s$ \\ \hline 
		$\Psi_s$ & Latency requirement of an SFC $s$ \\ \hline
		$\psi_v$ & Mean response time of a VNF $v$ \\ \hline 
		$D_s$ & Mean response time of an SFC $s$ \\ \hline 
		$b_v$ & Number of dedicated backups of the VNF $v \in V$ \\ \hline
		$b_c$ & Number of dedicated SFC backup chains \\ \hline  
		$l_c$ & Number of subchains \\ \hline
		$\lambda_s$ & Arrival rate of an SFC $s$ \\ \hline
		$\mu_v$ & Processing rate of a VNF $v$ \\ \hline 
		$pl(s)$ & Preference list of an SFC $s$ \\ \hline 
		$pl(n)$ & Preference list of a node $n$ \\ \hline 	
	\end{tabular}
\end{table}

In NFV environment, availability of a component (VNF or physical node) is defined as the ratio of the mean time the
component is up for delivering services to the sum of the mean time the component is up for delivering services and the mean time the component is down for repairing. Formally, it is defined as follows \cite{NFV_REL_003} \cite{Qu1_2018}: 
\begin{equation}
	\text{Availability} = \frac{MTBF}{MTBF+MTTR}~,
\end{equation} 
where MTBF stands for mean time between failures of the component and MTTR stands for mean time to repair the failed component. Reliability of a component (VNF or physical node) is defined as the probability that the component is available for providing services without failure for a stated period of time. 

Consider an SFC $s \in \mathcal{S}$, which is placed onto the substrate nodes as shown in Figure \ref{Fig:scb}. Let $\mathcal{N}_s$ denotes the set of all such substrate nodes in which VNFs of an SFC $s$ is placed. If each VNF $v \in \mathcal{V}_s$ in the SFC is reliable with a probability $p_v$ and each substrate node $n \in \mathcal{N}_s$ where VNFs are placed is reliable with the probability $p_n$, then the overall reliability of the SFC chain $r_s$ is calculated as,
\begin{align}
r_s = \prod_{v \in \mathcal{V}_s} p_v \times \prod_{n \in \mathcal{N}_s} p_n 
\label{eq:relSFC} 
\end{align} 

We use $c_v$ to denote the CPU resource requirement for each VNF $v \in \mathcal{V}_s$. The resource requirement $c_s$ of the entire service chain $s \in \mathcal{S}$ can be calculated as,
\begin{align}
c_s = \sum\limits_{v \in \mathcal{V}_s} c_v \label{eq:resReqSFC}
\end{align}

We consider that any service request $s \in \mathcal{S}$ has the reliability and the latency requirements denoted by $\Delta_s$ and $\Psi_s$, respectively. Hence, the service chain $s \in \mathcal{S}$ should satisfy the following conditions,
\begin{equation}
\sum\limits_{v \in \mathcal{V}_s} \psi_v \le \Psi_s \\
\end{equation}
\begin{equation}
r_s \geq \Delta_s
\end{equation}

To analyze the delay, we model every VNF in a chain as an M/M/1 queue and the entire SFC as tandem of M/M/1 network of queues. By Burke's theorem, the arrival rate $\lambda_s$ is same for all the VNFs in tandem of M/M/1 network of queues. The average response time of the SFC can be calculated as,
\begin{equation}
D_s = \sum \limits_{v \in \mathcal{V}_s} \psi_v = \sum\limits_{v \in \mathcal{V}_s}  \frac{1}{\mu_v - \lambda_s} 
\label{eq:delaySFC}
\end{equation}

\begin{figure}[t]
	\centering
	\begin{subfigure}{0.45\textwidth} 
		\vspace{-1.2cm}
		\includegraphics [width=\textwidth]{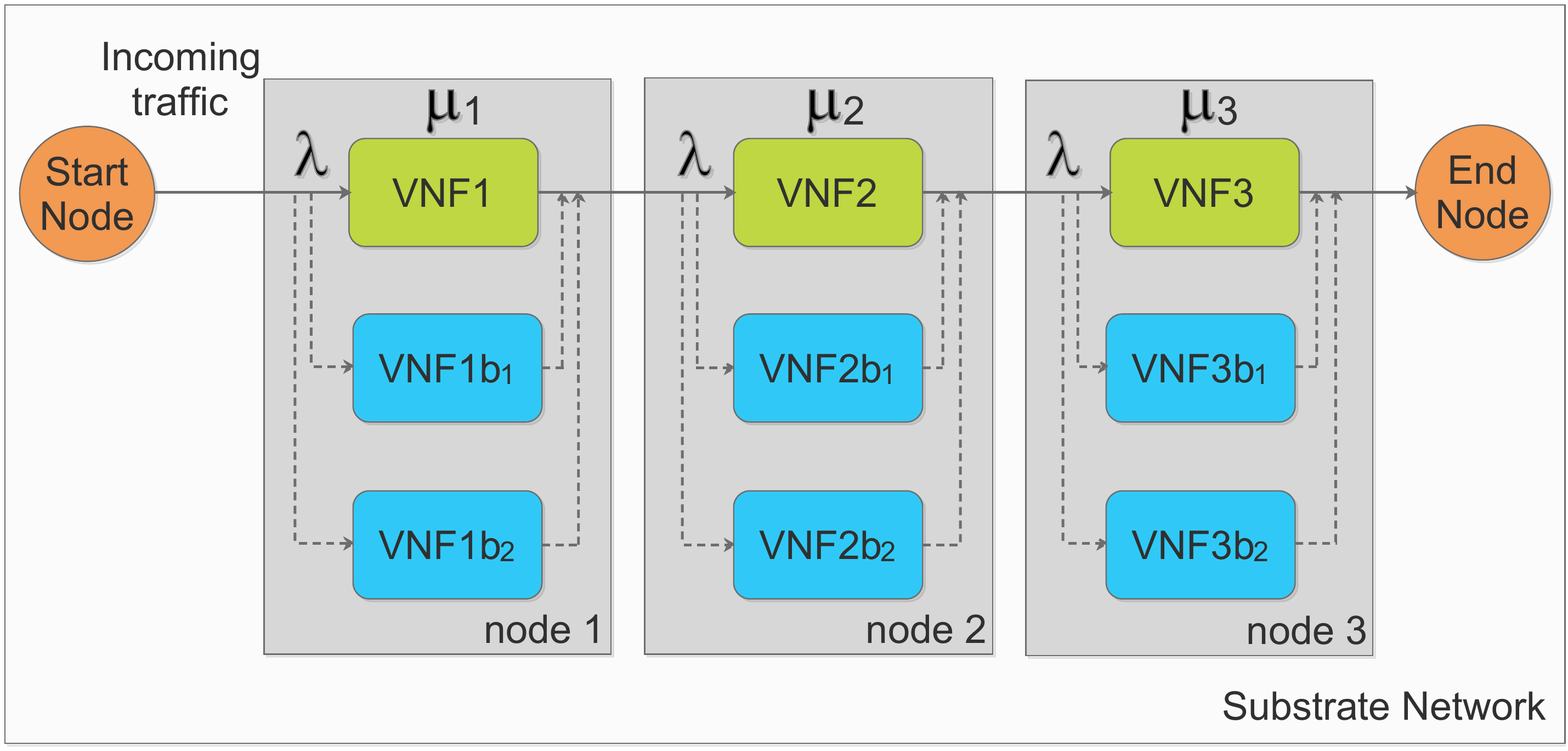}
		\vspace{-1.5cm}
		\caption{An SFC with dedicated backup VNFs.} 				
		\label{Fig:scb1}
		\vspace{-.7cm}
	\end{subfigure}
	\begin{subfigure}{0.45\textwidth}
		\includegraphics [width=\textwidth]{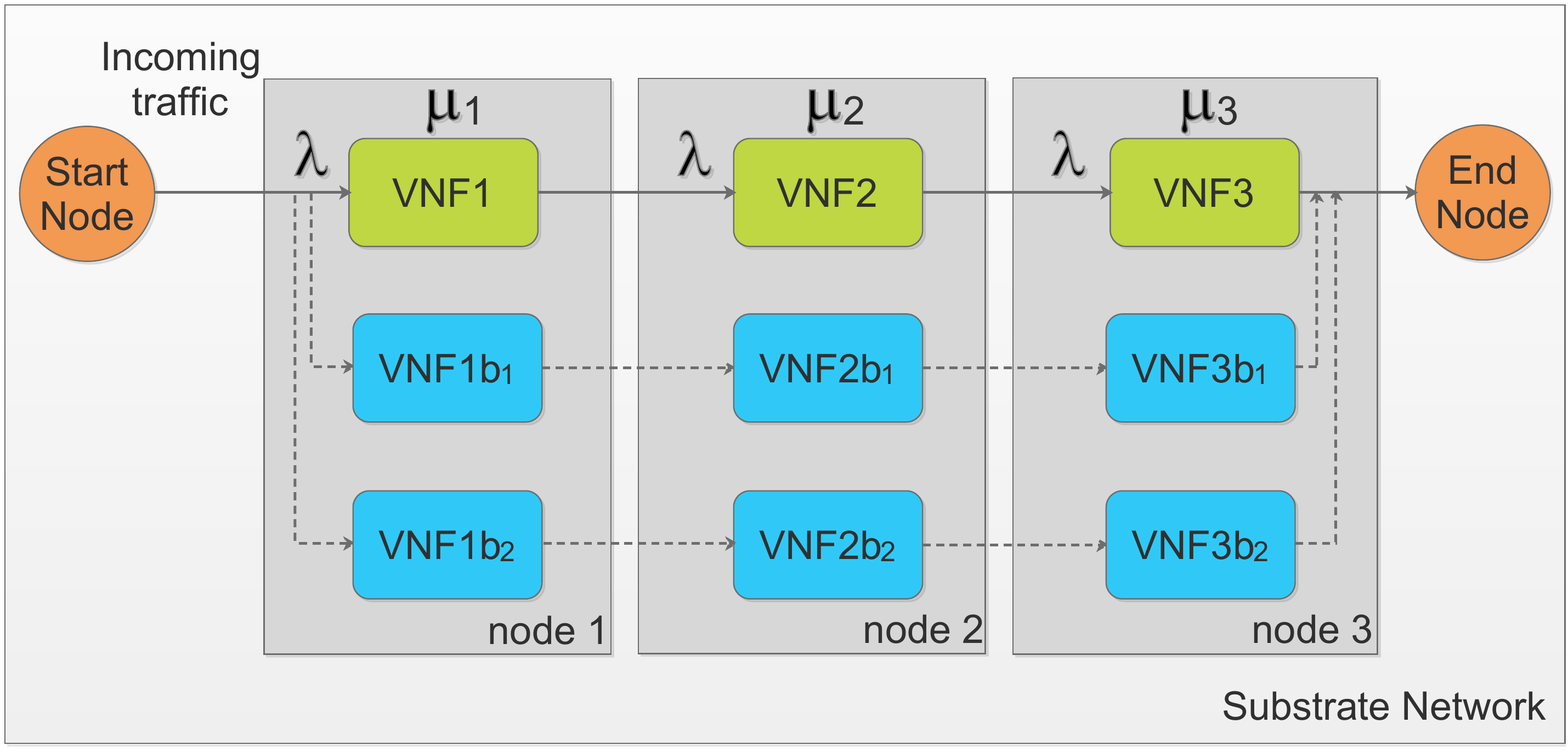}
		\vspace{-1.6cm}
		\caption{An SFC with separate chains as backup.} 				
		\label{Fig:scb2}
	\end{subfigure}
	\caption{An SFC with different backup settings.}
	\label{Fig:scb}
	
\end{figure}

Backups for an SFC can be added in two ways. One way is to assign a dedicated backup for each VNF of an SFC as shown in Figure \ref{Fig:scb1}. If any primary VNF component of the SFC fails, then the corresponding backup VNF is activated. The reliability of the SFC with dedicated backup VNFs can be calculated as,
\begin{equation}
r_{s_{b1}} = \prod_{v \in \mathcal{V}_s} \Big(1 - (1 - p_v)^{b_v + 1}\Big) \times \prod_{n \in \mathcal{N}_s} p_n
\label{eq:b1}  
\end{equation}
where $b_v$ is the number of dedicated backups of the same VNF type $v \in \mathcal{V}_s$ ($b_v \geq$ 1), and $n \in \mathcal{N}_s$ is the set of substrate nodes in which the VNFs (primary and its corresponding backup) are placed. We consider that a primary VNF and its corresponding backup VNFs are placed in the same substrate node to minimize the resource usage. VNFs can be placed in different substrate nodes with additional resource consumption. 

The resource requirement for the service chain with dedicated backups can be calculated as,
\begin{align}
c_{s_{b1}} = \sum\limits_{v \in \mathcal{V}_s} (b_v+1) \times c_v
\end{align}

The other way of adding backups for service continuity is to assign an entire SFC as backup as shown in Figure \ref{Fig:scb2}. If a VNF component of a primary SFC fails, then backup SFC will be activated. The reliability of SFC with separate SFC chains as backup can be calculated as,
\begin{equation}
r_{s_{b2}} = \Big(1 - (1 -  \prod\limits_{v \in \mathcal{V}_s} p_v)^{b_c+1}\Big) \times \prod \limits_{n \in \mathcal{N}_s} p_n  
\label{eq:b2}
\end{equation}
where $b_c$ is the number of separate SFC backup chains. 

The resource requirement for the service chain with separate SFC chains as backups can be calculated as,
\begin{align}
c_{s_{b2}} = (b_c+1) \times \sum\limits_{v \in \mathcal{V}_s} c_v
\end{align}

From Equations \eqref{eq:relSFC}, \eqref{eq:b1}, and \eqref{eq:b2}, it can be inferred that the service chain with backups has higher reliability than the one which does not have any backups. However, service chain with backups consumes more amount of resource in order to enhance the reliability. Hence, this approach is inefficient with respect to utilization of resources. The redundant backup resources are idle until a failure happens in the primary VNFs. Also, since failure may happen randomly at any point of time, assigned redundant backup resources cannot be used for any other purpose. 

\subsection{Enhancing Reliability of SFCs without Backups}
To efficiently utilize the available resources and enhance the reliability of service chains without assigning backups, we propose to divide the VNFs of an SFC into multiple lesser capacity VNFs and place them onto the substrate nodes. VNFs of an SFC can be divided into lower capacity VNFs and chained in parallel (rather than assigning dedicated backups in parallel) as shown in Figure \ref{Fig:subchaining}. We call this method as subchaining.  A reduced capacity (processing rate) VNF performs the same software functionality as that of the original VNF, and the reliability of each VNF is still $p_v$.

\begin{figure}[t]
	\centering
	\begin{subfigure}{0.45\textwidth} 
		\vspace{-1.9cm}
		\includegraphics [width=\textwidth]{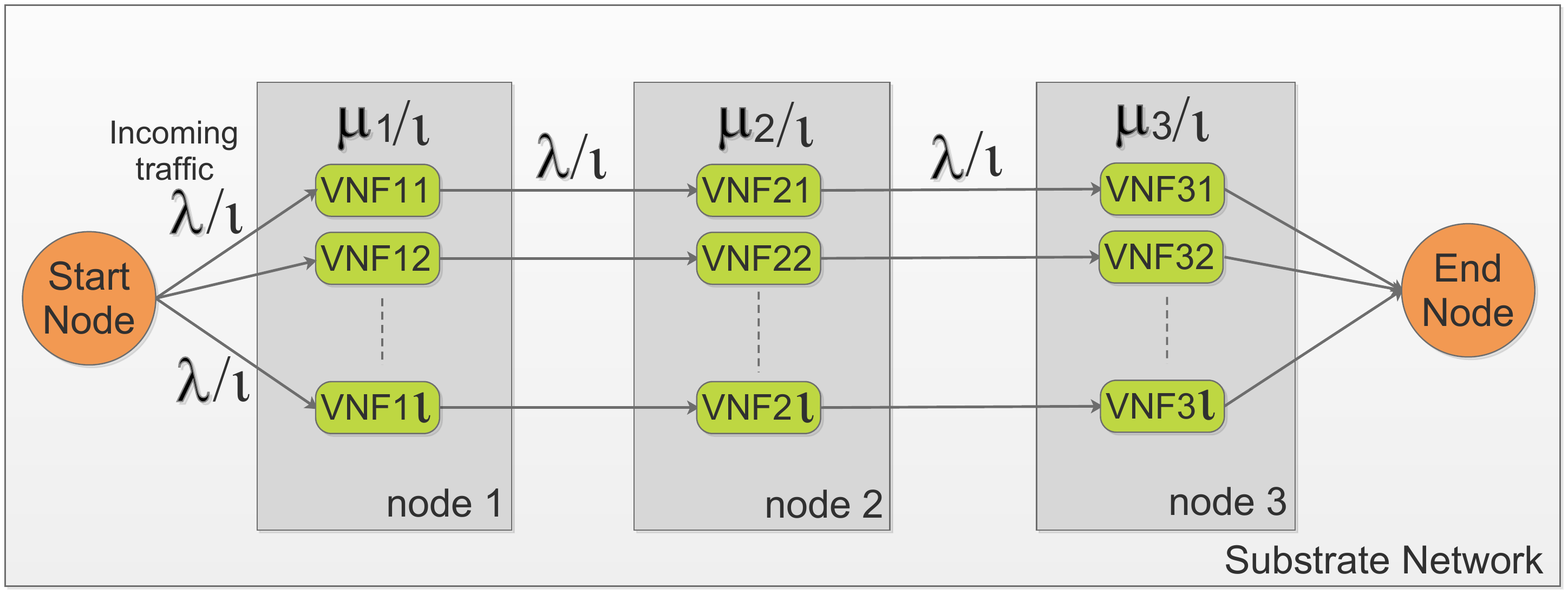}
		\vspace{-1.3cm}
		\caption{Subchaining an SFC as tandem of M/M/1 network of queues.} 				
		\label{Fig:subc1}
	\end{subfigure}
	\begin{subfigure}{0.45\textwidth}
		\vspace{-1.3cm}
		\includegraphics [width=\textwidth]{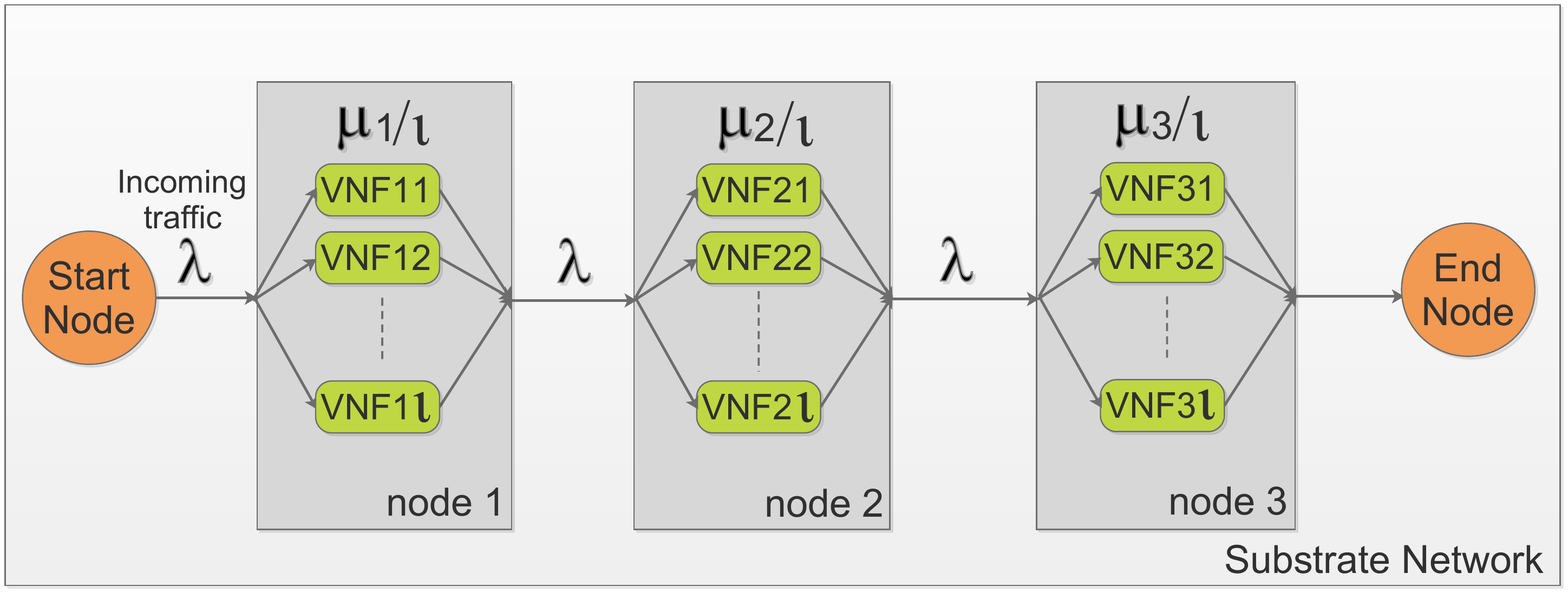}
		\vspace{-1.7cm}
		\caption{Subchaining an SFC as tandem of M/M/m network of queues.} 				
		\label{Fig:subc2}
	\end{subfigure}
	\vspace{-.2cm}
	\caption{An SFC subchaining methods.}
	\label{Fig:subchaining}
	
\end{figure}

If we divide an SFC into $l_c$ number of subchains with each VNF $v \in \mathcal{V}_s$ having processing capacity of $\frac{\mu_v}{l_c}$, then the incoming traffic is equally divided to all the subchains ($\frac{\lambda_s}{l_c}$). Each subchain of the SFC can be modeled as tandem of M/M/1 network of queues as shown in Figure \ref{Fig:subchaining}a. The reliability and average response time of $l_c$ subchains of the SFC in the tandem of M/M/1 queueing network setting can be calculated as,
\begin{align}
r_{s_{M/M/1}} = \Big(1 - (1 -  \prod\limits_{v \in \mathcal{V}_s} p_v)^{l_c} \Big) \times \prod \limits_{n \in \mathcal{N}_s} p_n
\label{eq:relMM1}
\end{align}
\begin{equation}
D_{s_{M/M/1}} = \sum\limits_{v \in \mathcal{V}_{\bar{s}}} \frac{1}{\frac{\mu_v}{l_c} - \frac{\lambda_s}{l_c}} = \sum\limits_{v \in \mathcal{V}_{\bar{s}}} \frac{l_c}{\mu_v - \lambda_s} \label{eq:delayMM1}
\end{equation}  

If we divide each VNF of an SFC into $l_c$ number of lesser capacity VNFs with processing capacity of $\frac{\mu_v}{l_c}$, then the incoming traffic $\lambda_s$ can be processed by any of the $l_c$ number of VNFs. It can be modelled as tandem of M/M/m network of queues as shown in Figure \ref{Fig:subchaining}b. The reliability and average response time of $l_c$ subchains of the SFC in this M/M/m setting can be calculated as,  
\begin{align}
r_{s_{M/M/m}} =  \prod_{v \in \mathcal{V}_s}\Big(1 - (1 - p_{v})^{l_c}\Big) \times \prod \limits_{n \in \mathcal{N}_s} p_n
\label{eq:relMMm}
\end{align}

\begin{align}
D_{s_{M/M/m}} = \sum\limits_{v \in \mathcal{V}_s} \frac{l_c}{\mu_v} \times \bigg(1+\frac{\varrho}{l_c(1-\frac{\lambda_s}{\mu_v})}\bigg)
\label{eq:delayMMm}
\end{align}    
where,
\begin{align*}
\varrho = \frac{(\frac{l_c\lambda_s}{\mu_v})^{l_c}}{{l_c}!(1-\frac{\lambda_s}{\mu_v})} \times \bigg( \frac{1}{1+\frac{(\frac{l_c\lambda_s}{\mu_v})^{l_c}}{{l_c}!(1-\frac{\lambda_s}{\mu_v})}+\sum \limits _ {i=1}^{l_c-1} \frac{(\frac{l_c\lambda_s}{\mu_v})^i}{i!}} \bigg)
\end{align*}

Dividing a single chain SFC into multiple subchains of SFC with lesser capacity VNFs enhances the reliability of the service chain without using backups \cite{PKT_2019}. At the same time, as shown in Equations \eqref{eq:delayMM1} and \eqref{eq:delayMMm}, dividing the VNFs of SFC into lesser capacity VNFs of multiple subchains increases the response time linearly. Therefore, SFC subchaining should be done without violating the delay constraint $\Psi_s$ of service request in the process of enhancing the reliability without backups.  

\subsection{Guaranteeing the Reliability Requirement of SFCs}
The number of SFC subchains that can be created to enhance the reliability of an SFC depends on the maximum allowed delay $\Psi _s$ of a service request. Hence, enhancing the reliability by SFC subchaining may not be sufficient to meet the reliability requirements of all the service requests. In such cases, in addition to the subchaining, redundant backups are added one by one to the less reliable VNFs to guarantee the reliability requirement of service requests. Methods for calculating reliability and guaranteeing the reliability requirement of service requests are given in Algorithms \ref{algo:alg1} and \ref{algo:alg2}, which utilize the SFC subchaining technique to reduce the number of redundant backup resources required to guarantee the reliability requirement. In this work, we assume that the links between physical nodes and virtual nodes are completely reliable.

\begin{algorithm}[t]
	\footnotesize
	\caption{The reliability calculation and subchaining procedure}
	
	\label{algo:alg1}
	\hspace*{\algorithmicindent} \textbf{Input:} $G_s = (\mathcal{V}_s, \mathcal{E}_s), \Psi_s, \Delta_s$ \\
	\hspace*{\algorithmicindent}  \hspace{-.2cm}\textbf{Output:} Reliability of an SFC $r_s$, and number of SFC subchains $l_1$ and \hspace*{1.3cm} $l_2$ are created for M/M/1 and M/M/m settings, respectively
	\begin{algorithmic}[1]
		\State Calculate the reliability of an SFC $r_s$ using Equation \eqref{eq:relSFC}
		\State Initialize the number of subchains $l_1$ = 1 and $l_2$ = 1 
		\If {$r_s \ge \Delta_s$} 
		\State Return $r_s$, and $l_1$ and $l_2$
		\Else 
		\While{$r_s < \Delta_s$}
		\If {M/M/1 setting}
		\State $l_1 = l_1 + 1$
		\State Create $l_1$ number of subchains from an SFC with the \hspace*{1.3cm}resource capacity of $\frac{c_v}{l_1}, ~ \forall v \in \mathcal{V}_s$
		\State Compute the overall delay of new subchains $l_1$ using \hspace*{1.3cm}Equation \eqref{eq:delayMM1} and assign it to $D_s$ 
		\If {$D_s \le \Psi_s$}
		\State Compute the reliability of new subchains using \hspace*{1.7cm}Equation \eqref{eq:relMM1} and assign it to $r_s$
		
		\Else 
		\State $l_1 = l_1 - 1$
		\State Return $r_s$ and $l_1$ 
		\EndIf
		\ElsIf{M/M/m setting}
		\State $l_2 = l_2 + 1$
		\State Create each VNF of an SFC into $l_2$ number of replicas \hspace*{1.3cm}with the resource capacity of $\frac{c_v}{l_2}$ 
		\State Compute the overall delay of new subchains $l_2$ using Equation \hspace*{1.2cm} \eqref{eq:delayMMm} and assign it to $D_s$
		\If {$D_s \le \Psi_s$}
		\State Compute the reliability of new subchains using \hspace*{1.7cm}Equation \eqref{eq:relMMm} and assign it to $r_s$
		
		\Else 
		\State $l_2 = l_2 - 1$
		\State Return $r_s$ and $l_2$
		\EndIf
		\EndIf 
		\EndWhile 
		\EndIf
	\end{algorithmic}
\end{algorithm}

\begin{algorithm}[t]
	\footnotesize	 
	\caption{The reliability calculation and reliability requirement guaranteeing procedure}
	
	\label{algo:alg2}
	\hspace*{\algorithmicindent} \textbf{Input:} $G_s = (\mathcal{V}_s, \mathcal{E}_s), \Psi_s, \Delta_s, r_s, l_1, l_2, p_n$ \\
	\hspace*{\algorithmicindent}  \hspace{-.1cm}\textbf{Output:} Guarantees the reliability requirement $\Delta_s$ for the service \hspace*{1.7cm}request $s \in \mathcal{S}$
	\begin{algorithmic}[1]	
		\If {$r_s \ge \Delta_s$}
		\State Redundant backup is not required
		\Else 
		\State Sort the VNFs of an SFC with respect to reliability in ascending order
		\State Start assigning backups from least reliable VNF to highest reliable \hspace*{0.4cm}VNF in a circular manner 
		\State $r_s$ is the reliability of an SFC computed by the subchaining procedure \hspace*{0.4cm}using Algorithm \ref{algo:alg1} 
		\State $l_1$ and $l_2$ are number of subchains created by subchaining procedure \hspace*{0.4cm}of M/M/1 and M/M/m settings, respectively 
		\State $p_n$ is the reliability of substrate nodes			
		\State Initialise $u$ = 2 (initially, one backup is assigned along with primary \hspace*{0.4cm}VNFs of a subchain), $w$ = 0 (number of subchains in which the same \hspace*{0.4cm}number of backups assigned to all VNFs of subchains), $j_1=0$, $j_2 = 0$
		\State Let $Q$ be the set of VNFs for which backup is assigned and initially \hspace*{0.4cm}it is null 
		\While {$r_s < \Delta_s$} 
		\If {M/M/1 setting}
		\State $Q = {Q \cup \{\arg \min \limits_{v \in {\mathcal{V}_s - Q}} p_v\}}$
		\State $h_1 = \prod \limits_{v \in \mathcal{V}_s}\bigg(1-(1-p_v)^u\bigg)$ 		\State $h_2 = \Bigg(\prod \limits_{v \in Q}\bigg(1-(1-p_v)^{u}\bigg) ~~\quad\times \hspace*{5cm}\prod 
		\limits_{v \in \mathcal{V}_s - Q }
		\Big(1 - (1 - p_v)^{u - 1}\Big) \Bigg)$ 		\State $h_3 = \prod \limits_{v \in \mathcal{V}_s}\Big(1 - (1 - p_v)^{u - 1}\Big)$ 
		\State Reliability $r_s = \Bigg(\bigg(1 - \Big((1-h_1)^w \times (1 - h_2) \times (1- h_3)^{l_1 - 1 - w}\Big)\bigg) \times \hspace*{7cm}\prod \limits_{n \in \mathcal{N}_s} p_n\Bigg)$
		\State $j_1$ = $j_1$ + 1
		\If {$j_1 == |\mathcal{V}_s|$}
		\State $j_1$ = 0, $w = w + 1$, $Q = \{\}$
		\EndIf 
		\If {$w$ == $l_1$}
		\State $w = 0, ~u = u + 1$
		\EndIf 
		\ElsIf {M/M/m setting}
		\State $Q = {Q \cup \{\arg \min \limits_{v \in {\mathcal{V}_s - Q}} p_v\}}$
		\State Reliability $r_s = \Bigg(\prod \limits_{q \in Q} \bigg(1-(1-p_q)^{l_2+1}\bigg) \times \hspace*{4cm}\prod 		\limits_{i \in \mathcal{V}_s - Q }
		\bigg(1-(1-p_i)^{l_2}\bigg) \times \prod \limits_{n \in \mathcal{N}_s} p_n\Bigg)$   
		\State $j_2 = j_2 + 1$
		\If {$j_2 == |\mathcal{V}_s|$}
		\State $l_2 = l_2 + 1, j_2 = 0, \text{and} ~Q = \{\}$
		\EndIf  
		\EndIf 
		\EndWhile 
		\EndIf 
	\end{algorithmic}
\end{algorithm}

The reliability calculation and subchaining procedure is given in Algorithm \ref{algo:alg1}. For the given input of service request graph and its requirements, Algorithm \ref{algo:alg1} applies subchaining procedure and outputs the reliability of an SFC chain along with the number of subchains created. First, the reliability of an SFC chain $r_s$ is calculated using Equation \eqref{eq:relSFC} and returned if it satisfies the requirement $\Delta_s$ (lines 1 to 4). If the reliability value $r_s$ is less than the requirement $\Delta_s$, then subchain count $l_1$ (or $l_2$) is increased and the SFC chain is divided into $l_1$ (or $l_2$) number of subchains in each iteration to enhance the reliability. The process continues till either the reliability requirement is met or the maximum delay constraint $\Psi_s$ is violated (lines 6 to 16 for M/M/1 and lines 17 to 28 for M/M/m setting). The maximum number of subchains that can be created is limited by the delay constraint $\Psi_s$. Algorithm 1 always finds a solution in finite number of iterations irrespective of the number of SFC subchains created. Because either average response time of the latest subchains of SFC $D_s$ violates the delay constraint $\Psi_s$ or the improved reliability of the latest subchains of SFC $r_s$ meets the reliability requirement $\Delta_s$ . Therefore, Algorithm \ref{algo:alg1} terminates in a finite number of iterations.

Algorithm \ref{algo:alg2} is designed to guarantee the reliability requirement of service requests for the cases the subchaining procedure described in Algorithm \ref{algo:alg1} could not satisfy the reliability requirement. For the given input of service request graph and its requirements and the output of the Algorithm \ref{algo:alg1}, Algorithm \ref{algo:alg2} outputs the service graph which guarantees the reliability requirement by adding backups in incremental manner. If the output of Algorithm \ref{algo:alg1} satisfies the reliability requirement, then backups are not required (lines 1 to 3). First, the VNFs of an SFC are sorted with respect to reliability value in ascending order to provide backups in a circular manner from least reliable VNF to highest reliable VNF of a subchain and the input parameters are initialized (lines 4 to 10). At any point, backup is added to only one VNF of a subchain in order to reduce the number of redundant resources (line 13 for M/M/1 and 26 for M/M/m). Since backup is added one by one, the SFC can consist of three different subchain structures in M/M/1 setting as shown in Figure \ref{Fig:sc1}: i) backup is already assigned to all the VNFs in the subchain (line 14), ii) backup is assigned to part of the VNFs of the subchain and backup is not yet added to the remaining VNFs of the subchain (line 15), and iii) backup is not yet added to any of the VNFs in the subchain (line 16). Since the subchains are in parallel and multiple subchains can have the same backup structure, the reliability of the whole chain is computed based on the three backup structures reliability values ($h1,h2,h3$) and reliability of the substrate nodes ($p_n, \forall n \in \mathcal{N}_s$) where the subchains are placed (line 17). In M/M/m setting, two different subchain structures are possible while adding backups as shown in Figure \ref{Fig:sc2}: i) the VNFs to which a backup is added (i.e., such VNF has an additional redundant VNF), ii) the VNFs to which no backups are added (line 27). In both M/M/1 and M/M/m settings, backups are added one by one until the reliability requirement is met (lines 18 to 24 for M/M/1 and 28 to 33 for M/M/m). Algorithm \ref{algo:alg2} always finds a solution in finite number of iterations irrespective of the number of redundant VNFs assigned. This is because the reliability of an SFC is increasing after each iteration by adding a new backup VNF to the least reliable VNF in the SFC chain. Therefore, Algorithm \ref{algo:alg2} terminates in
a finite number of iterations in polynomial time. The average running times of Algorithms \ref{algo:alg1} and \ref{algo:alg2} are given in Table \ref{tab: running_time} in Section VI B. 
 
\begin{figure}[]
	\centering
	\begin{subfigure}{0.45\textwidth} 
		\vspace{-.5cm}
		\includegraphics [width=\textwidth]{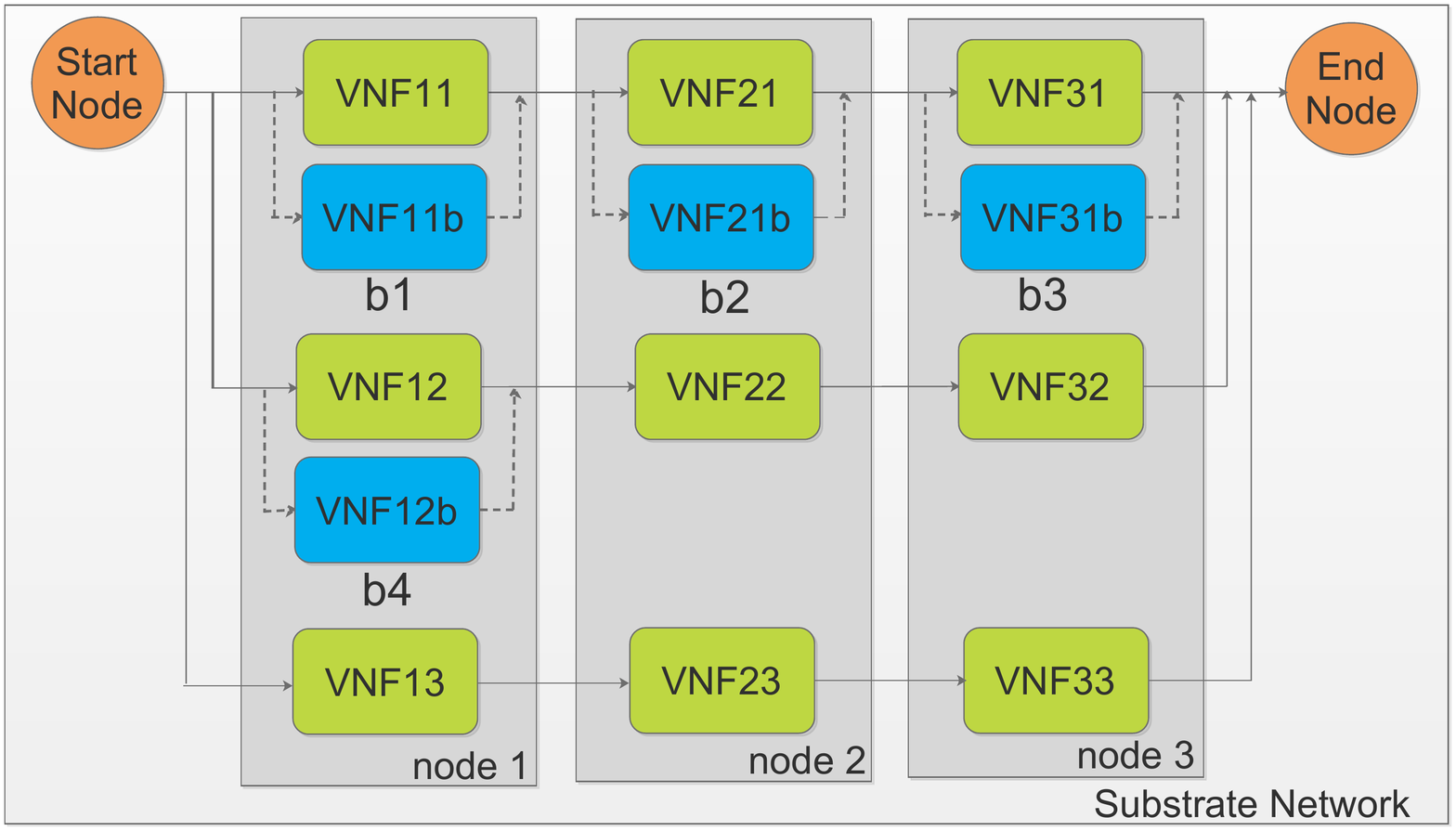}
		\vspace{-1.7cm}
		\caption{Backup structures for an SFC with M/M/1 tandem network of queues.} 				
		\label{Fig:sc1}
	\end{subfigure}
	\begin{subfigure}{0.45\textwidth} 
		\vspace{-0.5cm}
		\includegraphics [width=\textwidth]{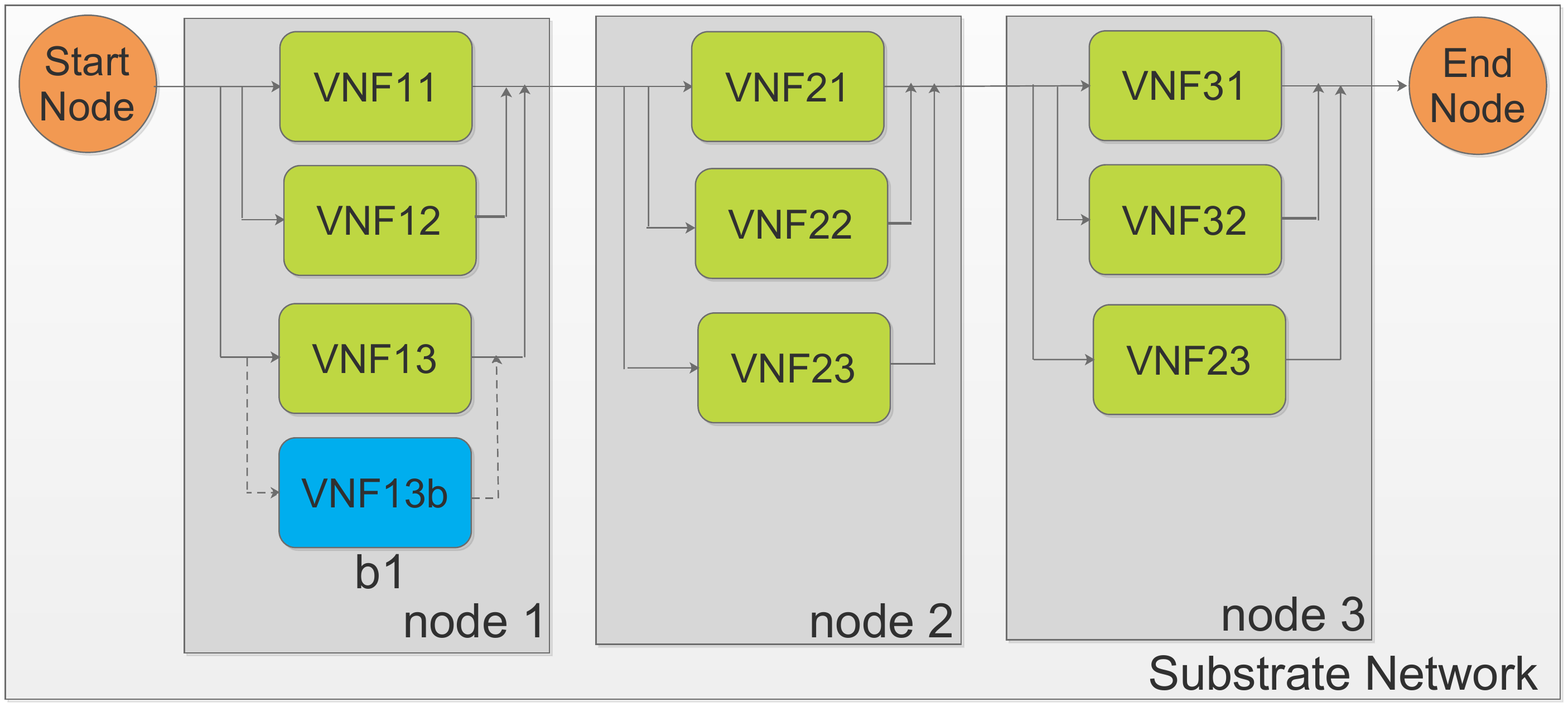}
		\vspace{-2cm}
		\caption{Backup structures for an SFC with M/M/m tandem network of queues.} 				
		\label{Fig:sc2}
	\end{subfigure}
	\caption{Subchaining and backup structures of an SFC.}
\end{figure}

\section{Placement of Reliable SFCs}
\label{placement}
This section describes about optimal placement of the designed reliable SFC graphs in the NFV based 5G infrastructure with minimal number of physical resources. Optimal on-demand dynamic resource allocation is crucial to provide diverse set of communication services using SFCs in 5G networks, and also reduce CAPEX and OPEX. 
SFC placement or resource allocation problem is the process of mapping SFCs to physical/substrate network in optimal manner while meeting the SLAs. First, we mathematically model the reliable SFC placement problem using ILP and prove that the problem is NP-hard. Then, we propose a modified matching algorithm for solving large scale instances of the problem in polynomial time.

\subsection{ILP Mathematical Formulation}
\begin{enumerate}
	\item Decision variables: 
	The binary variable $x_{ns}$ is used to represent that all VNFs of an SFC $s$ are placed on the processing node $n$, which can be expressed as,  
	\begin{equation}
	x_{ns} = \begin{cases} 1, &\text{if all VNFs of the SFC $s \in \mathcal{S}$ are placed} \\&\text{on the processing node $n \in \mathcal{N}$}, \\ 
	0, & \text{otherwise}. \end{cases} 
	\end{equation} 
	
	We assume that all subchained VNFs from the original VNF should be placed on the same physical node because it reduces the consumption of physical resources to minimize operational expenditures for network operators. Moreover, placing the VNFs of an SFC on the same physical node reduces switchover time, amount of bandwidth consumed to transfer VNF internal state information from primary to backup VNFs, and inter VNF communication delays.
	
	The binary variable $a_n$, used to represent that  a processing node $n$ is active, which can be expressed as,
	\begin{equation}
	a_n = \begin{cases} 1, &\text{if the processing node $n \in \mathcal{N}$ is active, i.e.,} \\ &\text{if it hosts at least one SFC $s \in \mathcal{S}$}, \\ 
	0, & \text{otherwise}. \end{cases} 
	\end{equation} 
	
	\item 
	Objective function: The objective is to minimize the number of active physical nodes allocated for SFCs deployment, which can be expressed as,
	\begin{equation}
	Z: \text{min} \sum_{n \in \mathcal{N}} a_n
	\end{equation}
		
	\item 
	Capacity constraint: The capacity requirements of SFCs placed on any physical node should not exceed that server's available resource capacity, which can be mathematically expressed as,
	\begin{equation}
	\sum_{s \in \mathcal{S}} x_{ns} \times c_s \leq c_n \times a_n, \forall n \in \mathcal{N}
	\end{equation}
	where $c_n$ denotes the available CPU resource capacity in the physical node $n$ and $c_s$ is the total CPU requirement of the SFC chain $s$ defined in Equation \eqref{eq:resReqSFC}. We assume that all substrate nodes have the same resource
	capacity to host VNFs of an SFC.

	In general, multiple resource type (e.g., CPU, memory, and storage) constraints can be modeled based on the service type requirements, which can be expressed as,
	 
	\begin{equation}
	\sum_{s \in \mathcal{S}} x_{ns} \times c_{sr} \leq c_{nr}, \forall n \in \mathcal{N}, \forall r \in \mathcal{R}
	\end{equation}
	where $r$ is a particular resource type (e.g., memory) from the resource type set $\mathcal{R}$ which includes CPU, memory, and storage. $c_{sr}$ denotes resource requirement of type $r$ for the SFC $s$ and $c_{nr}$ denotes the available capacity of resource type $r$ in the node $n$. For simplicity, we consider only virtual CPU resource requirement for VNFs to process the incoming traffic \cite{power_aware}.
	 
	\item 
	Placement constraint: The SFC should be placed in any one of the physical nodes only, which can be expressed as,
	\begin{equation}
		\sum_{n \in N} x_{ns} = 1, \forall s \in \mathcal{S}
	\end{equation}
	
\end{enumerate}

\begin{theorem} 
Reliable SFC placement problem $Z$ is NP-hard.
\end{theorem}
\begin{proof}
 Let A be the reliable SFC placement problem and B be the bin packing problem. Bin packing problem is one of the famous combinatorial optimization problems and it is an NP-hard problem \cite{Garey_Johnson}, which is defined as follows: given a set of items, each having an integer weight, and a set of identical bins each having an integer capacity $c$, the problem consists of packing all the items into minimum number of bins, without exceeding the maximum capacity $c$ for any bin. To prove that the problem A is NP-hard, it is sufficient to show that an instance of problem B can be reduced to an instance of problem A in polynomial time, i.e., B $\leq_P$ A~\cite{Cormen_2009}.

 We can transform an instance of problem B into an instance of problem A in the following way: i) consider each item in the bin packing problem as an SFC in SFC placement problem, ii) set the integer weight of each item to be equal to the resource requirement of each SFC iii) consider total number of available bins as total number of substrate nodes, iv) set the capacity of each bin to be equal to the resource availability in each substrate node, and v) consider that each item is placed in only one bin as an SFC is placed in only one of the substrate nodes. The transformation operation can be done in polynomial time of the input size. 
 
 Hence, problem B is reducible to problem A in polynomial time. If A is not NP-hard, then B is also not NP-hard (since B is reducible to A), which is a contradiction. Therefore, it can be concluded that A is also an NP-hard problem.        
\end{proof}

\subsection{Matching Algorithm Based Reliable SFC Placement}
SFC placement problem, being an NP-hard problem, takes super-polynomial time to solve when the input size is large. We devise a matching game based solution to overcome the computational complexity. A well-known Gale-Shapley \cite{Gale_Shapley} matching algorithm framework based on deferred acceptance concept is used to place SFCs onto the substrate nodes. There are three types of matching techniques: one-to-one matching, many-to-one matching, and many-to-many matching. Since multiple SFC chains can be placed on the same substrate node and an SFC is placed on only one substrate node, many-to-one matching technique is used in our solution design. Since each SFC may have different set of VNFs, and resource requirement for each SFC may vary for different service types, classical matching theory approach cannot be applied directly for resource-efficient SFC placement problem. Hence, we propose a modified matching algorithm to place SFCs on the substrate nodes efficiently. SFCs propose to  substrate nodes based on the preferences of SFCs in SFC-optimal stable matching procedure, whereas substrate nodes propose to SFCs based on the preferences of substrate nodes in substrate node-optimal stable matching procedure. We consider SFC-optimal stable matching in order to meet the SLAs of user/service requests.

\begin{figure*}[t]
	\centering
	\begin{subfigure}{0.47\textwidth} 
		\vspace{-1.3cm}
		\hspace{-1.5cm}
		\includegraphics [scale=0.37]{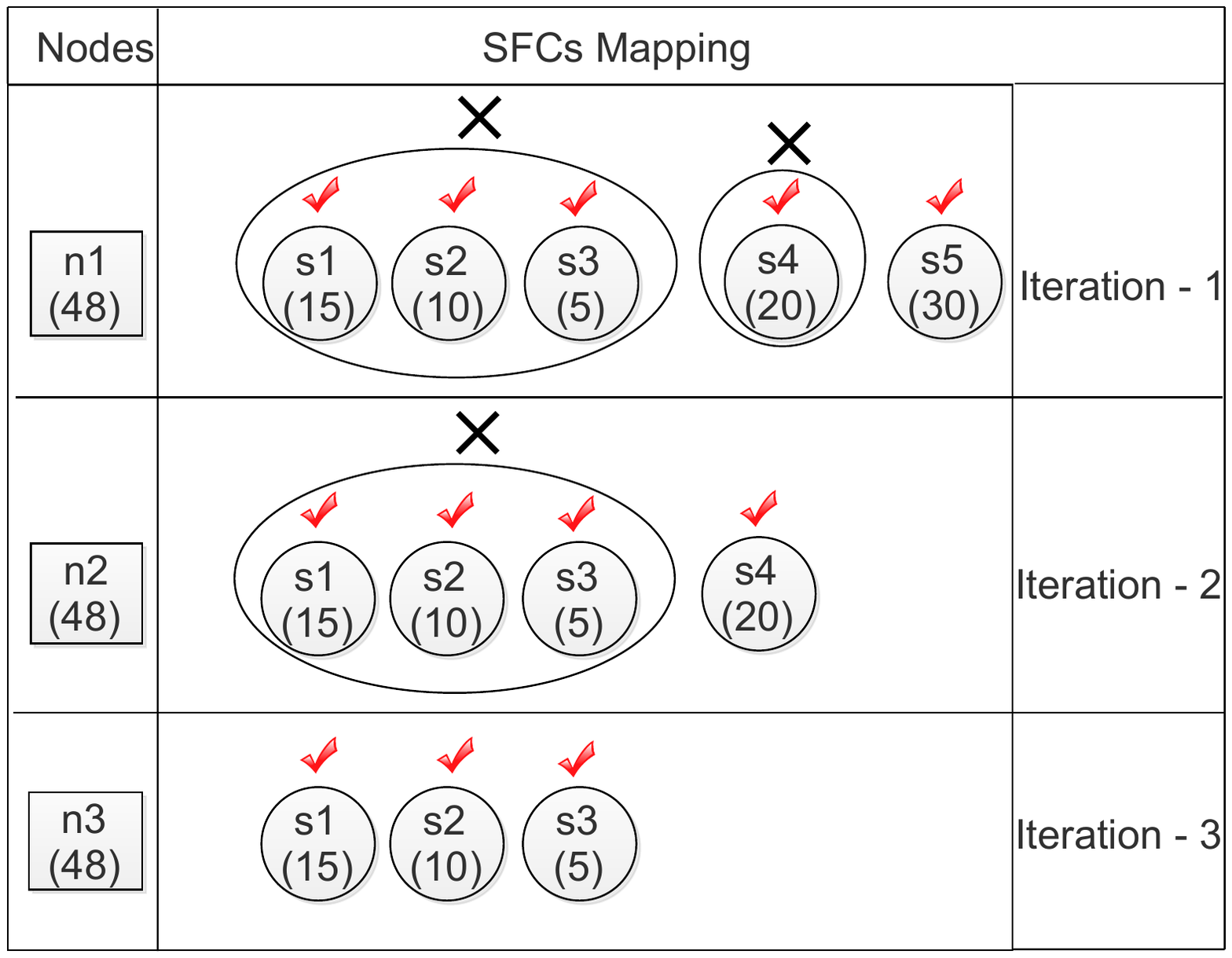}
		\hspace{-1.5cm}
		\vspace{-2.7cm}
		\caption{Many-to-one matching resource allocation strategy \cite{Pham_2018}.} 				
		\label{Fig:m1}
	\end{subfigure}
	\hspace{0.75cm}
	\begin{subfigure}{0.47\textwidth}
		\vspace{-1.8cm}
		\hspace{-2.2cm}
		\includegraphics [scale=0.4]{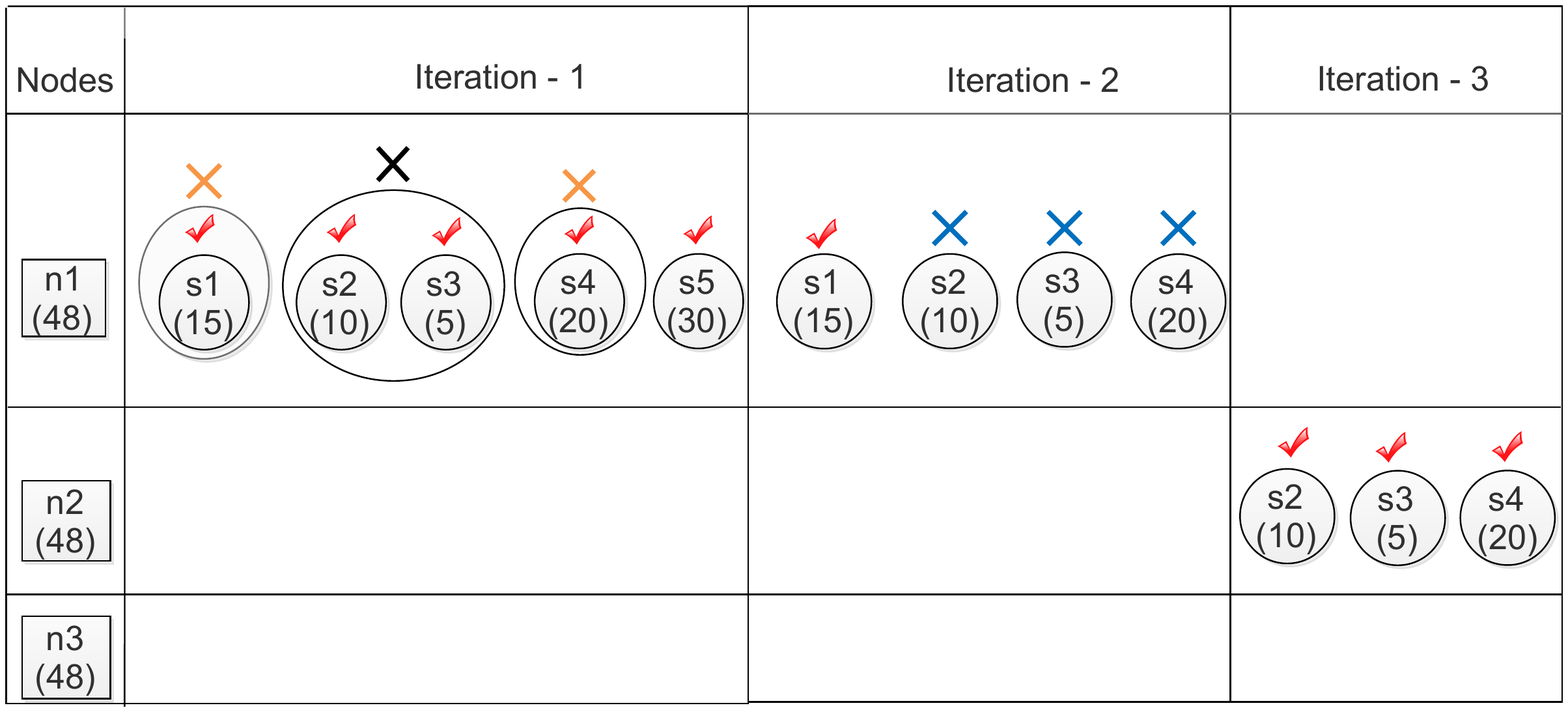}
		\hspace{-1.5cm}
		\vspace{-3.1cm}
		\caption{Proposed many-to-one matching resource allocation strategy.} 				
		\label{Fig:m2}
	\end{subfigure}
	\caption{Many-to-one matching resource allocation strategies.}
	\label{Fig:mm}
	\vspace{-.5cm}
\end{figure*}

\begin{figure} 
	\centering
	\vspace{-.2cm}
	\includegraphics [scale=0.3]{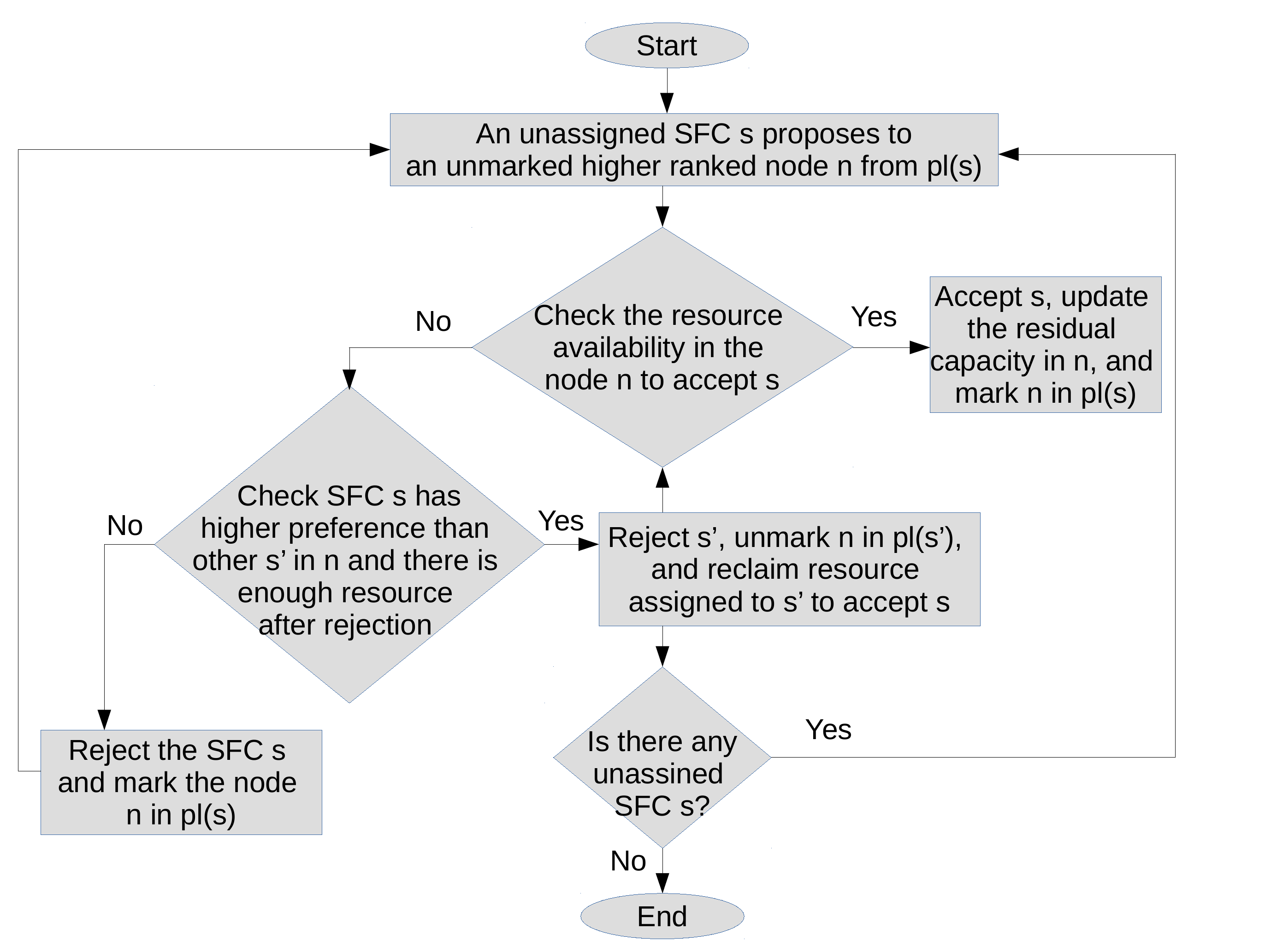}
	\caption{Resource-efficient assignment of SFCs to nodes.} 				
	\label{Fig:flowchart}
\end{figure}

A many-to-one matching game consists of two disjoint sets of groups with finite number of players $\mathcal{N}=\{n_1,n_2,\ldots,n_{|\mathcal{N}|}\}$ and  $\mathcal{S}=\{s_1,s_2,\ldots,s_{|\mathcal{S}|}\}$, where $\mathcal{N} \cap \mathcal{S} = \emptyset$. Each player prepares a preference list to match with the player in the other group. A preference list is the order of preference in which a player in one group ranks all the players in the other group based on some performance metric. We use the preference relation symbols $\succ_{s_j}$ and $\succ_{n_i}$ to denote the preference orderings of players $s_j \in \mathcal{S}$ and $n_i \in \mathcal{N}$, respectively. For example, $s_2 \succ_{n_1} s_1$ indicates that player $n_1$ gives higher preference to $s_2$ than to $s_1$. The preference list of a player $s_j$ is represented as $pl(s_j) = \{n_2,n_5,n_3\ldots,n_{|\mathcal{N}|}\}$, if player $s_j$'s first choice is $n_2$, second choice is $n_5$, and so on. Each player's preference list in both the groups should be strict, complete, and transitive. Strict preference relation indicates that each player $s_j \in \mathcal{S}$ has a strict preference relation $\succ_{s_j}$  over the set of players $n_i \in \mathcal{N}$ (i.e., no two players can be ranked with the same preference), and vice versa. Complete preference relation indicates that the preference list should include all the players in the opposite group. Transitive preference relation for a player $s_j$ indicates that if $n_2$ has higher preference than $n_5$ and $n_5$ has higher preference than $n_3$, then $n_2$ has higher preference than both $n_5$ and $n_3$. Individual rationality of all the players of both the groups is considered in this matching algorithm based approach.   
 
Matching is a function $\tau :\mathcal{S} \cup \mathcal{N} \to 2^{\mathcal{S} \cup \mathcal{N}}$ which maps from the set $\mathcal{S} \cup \mathcal{N}$ into the subsets of $\mathcal{S} \cup \mathcal{N}$ (in other words, every player of $\mathcal{S}$ is mapped to exactly one player of $\mathcal{N}$). We consider that the sets $\mathcal{N}$ and $\mathcal{S}$ indicate the set of substrate nodes and SFC requests, respectively. Each node belonging to the group $\mathcal{N}$ (i.e., $n_i \in \mathcal{N}$) has positive resource capacity $c_{n_i} \in \mathbb{Z}^+$ to accommodate multiple SFCs, say g number of SFCs, from the group $\mathcal{S}$. The SFCs $s_j \in \mathcal{S}$ have resource requirement of $c_{s_j} \in \mathbb{Z}^+ \leq c_n$. Multiple SFCs in $\mathcal{S}$ can be mapped to a single node of $\mathcal{N}$ based on the available resource (residual capacity). $\tau(s_j)$ represents a node $n_i \in \mathcal{N}$ to which an SFC $s_j \in \mathcal{S}$ is assigned to, and $\tau(n_i)$ represents an SFC $s_j \in \mathcal{S}$ assigned to a node $n_i \in \mathcal{N}$.  
A pair ($s_j, n_i$) is said to be acceptable pair iff both $s_j$ and $n_i$ prefer each other based on their preference lists, i.e., $\tau(n_i) = s_j$ and $\tau(s_j)= n_i$.

A pair is said to be a blocking pair ($s_j,n_i)$ if both of them prefer to be matched with each other rather than being matched according to matching function $\tau$. A matching can be blocked by an individual player as well if the player prefers being single to being matched with a partner from the preference list. If a matching $\tau$ is not blocked by an individual or pair, then it is said to be stable.   

In classical many-to-one stable resource allocation problem \cite{Pham_2018} \cite{Gale_Shapley}, if there is a request from higher ranked SFC and the available resource (residual capacity) in the node is not enough to accept the proposal, then all the lesser preferred accepted SFCs are rejected. It is done in order to accept the higher ranked SFC. In this method, the rejected SFCs are not allowed to propose again to the same node in the next iterations. It results in inefficient utilization of resources. We illustrate this with an example shown in Figure \ref{Fig:mm}.  We consider three substrate nodes (n1, n2, and n3) each with the capacity of 48 vCPUs and five SFCs (s1, s2, s3, s4, and s5) with capacity requirements of 15, 10, 5, 20, and 30 vCPUs, respectively. For an easy illustration of example, we consider that all SFCs have the same node preference list pl(s) which is (n1, n2, n3) and all nodes have the same SFC preference list pl(n) which is (s5, s4, s1, s2, s3). 

As shown in Figure \ref{Fig:mm} (a), initially, the node n1 accepts the proposals of s1, s2, s3 (total vCPUs requirement is 30) in iteration 1. When the node~n1 receives proposal from a higher ranked SFC s4 in iteration 1, since the total resource requirement of SFCs s1, s2, s3, and s4 exceeds the total capacity of node n1 (50 > 48), the node n1 rejects all the already accepted SFC proposals (s1, s2, s3) in order to accept the higher ranked SFC s4. Similarly, when s5 proposes to the node n1, s4 is rejected in favor of s5. Since rejected SFCs (s1, s2, s3, s4) are not allowed to propose to the node n1 again, they propose to the node n2 in iteration 2. In iteration 2 also, the SFCs s1, s2, and s3 are rejected by the node n2 in favor of s4 when capacity requirement exceeds the available resource limit. Hence, in iteration 3 the rejected SFCs propose to the node n3 and are accepted. This strategy requires three nodes to accommodate all the five SFCs. 

Consider an SFC $s_j$ is rejected by a node $n_i$ in iteration 1. In our design, we allow $s_j$ to propose again to the node $n_i$ in further iterations until either the available resource in $n_i$ is not enough to accommodate $s_j$ or $s_j$ is lesser preferred to the already accommodated SFCs in the node $n_i$. We also precompute the resource to be reclaimed when rejecting the lesser preferred SFCs before actually rejecting them. Hence, the already  accommodated lesser preferred SFCs are rejected only if the total (residual + reclaimed) resource is enough to accommodate the higher ranked SFC. Resource-efficient assignment of SFCs to nodes is shown in Figure \ref{Fig:flowchart}. 

As shown in Figure \ref{Fig:mm} (b), SFCs s1, s2, and s3 are accepted initially in iteration 1 since there is enough resource in the node n1. When s4 proposes to the node n1, the required resource amount exceeds the available resource capacity (50 > 48) to accommodate all SFCs s1, s2, s3, and s4. Since s4 has higher preference in the node n1's preference list than s1, s2, s3 and the additional capacity to be obtained by rejecting the lesser preferred SFCs is more than the required amount, the node n1 rejects the lesser preferred SFCs one by one till the resource requirement is met. Hence, only s3 and s2 are rejected in order to accommodate s4. At this point of time, s1 and s4 are accepted and s2 and s3 are rejected. Similarly, when s5 proposes to the node n1, s1 and s4 are rejected in favor of s5 in iteration 1. In iteration 2, the rejected SFCs propose again to node 1 (remaining capacity 18 vCPUs). First, s1 is accepted by node 1 (remaining capacity 3 vCPUs), then s2 and s3 are rejected by the node 1 due to lack of resource availability in node 1 (10 > 3, 5 > 3) and they are lesser preferred compared to s1 in the node~1's preference list. When s4 proposes to the node 1, the available remaining capacity is not sufficient (20 > 3) to accommodate both s1 and s4. Though the SFC s4 has higher preference than s1, rejecting the SFC s1 and reclaiming the assigned capacity of s1 is not sufficient (20 > 18) to accommodate s4. Therefore the SFC s4 is rejected by the node 1. At the end of iteration 2, only the SFC s1 is accepted and the remaining SFCs s2, s3, and s4 are rejected. Since s1, s2, and s4 are rejected due to lack of resource in the node 1, they propose to the node 2 in the third iteration and are being accepted because of enough resource availability. Our strategy requires only two nodes to accommodate all five SFCs and improves the overall resource utilization.  

\begin{algorithm}[t]
	\footnotesize
	\caption{Modified matching algorithm based reliable SFC placement}
	\label{algo:Matching algorithm}
	\hspace*{\algorithmicindent} \textbf{Input:} The set of SFCs $\mathcal{S}$ and the set of available nodes $\mathcal{N}$ in \hspace*{1.6cm}the substrate network \\
	\hspace*{\algorithmicindent}  \hspace{-.05cm}\textbf{Output:} SFC-optimal stable matching result is produced such \hspace*{1.8cm}that all SFCs are placed on the substrate network
	\begin{algorithmic}[1]
		\State Check resource availability at the substrate network		
		\State Prepare SFCs' preference list $pl(s_j), \forall s_j \in \mathcal{S}$
	    \State Prepare substrate nodes' preference list $pl(n_i), \forall n_i \in \mathcal{N}$
					
		\While $~\exists s_j \in \mathcal{S}$ and all the available nodes are not marked in \hspace*{1.4cm}its preference list $pl(s_j)$
		    \State $n_i$ $\leftarrow$ choose the most preferred unmarked node from the \hspace*{0.5cm}preference list $pl(s_j)$  
				
		    \If {$c_{n_i} \ge c_{s_j}$} 

      	      \For {$\forall v \in \mathcal{V}_{s_j}$}
		        \State Map $v$ to $n_i$
		      \EndFor

		      \State $c_{n_i}$ = $c_{n_i} - c_{s_j}$	
		      \State Mark SFC $s_j$ on the node $n_i$ preference list $pl(n_i)$
		      \State Mark node $n_i$ on the SFC $s_j$ preference list $pl(s_j)$

		    \Else

		     \If {SFC $s_j$ is most preferred than the already mapped \hspace*{1.3cm}(marked) SFCs $s^\prime_j \in \mathcal{S}$  on $pl(n_i)$, i.e., $s_j \succ_{n_i} s^\prime_j$, and \hspace*{1.3cm}$\big(\sum_{s^\prime_j \in \mathcal{S}} ~c_{s^\prime_j}\big) + c_{n_i} \ge c_{s_j}$}

		      \Repeat 
		        \State Reject the least preferred $s^\prime_j$ 		
		        \State $c_{n_i}$ = $c_{n_i} + c_{s^{\prime}_j}$		
		        \State Unmark node $n_i$ on the preference list $pl(s^\prime_j) $		
		        \State Unmark SFC $s^\prime_j$ on the preference list $pl(n_i)$
	         \Until {{$c_{n_i} \geq c_{s_j}$}}

			\ElsIf {$c_{n_i} <  c_{s_j}$}
		
				\State The node $n_i$ rejects the proposal from the SFC $s_j$ 	
				\State Mark node $n_i$ on the preference list $pl(s_j)$

			\EndIf

		 \EndIf

		\EndWhile 
				
	\end{algorithmic}
\end{algorithm}

Algorithm \ref{algo:Matching algorithm} provides the procedure for placing SFCs in substrate nodes based on the modified matching algorithm. The same procedure is given in flow chart format in Figure \ref{Fig:flowchart}. For the given input of the set of reliable SFC design graphs and physical network graph, Algorithm \ref{algo:Matching algorithm} outputs SFC-optimal stable matching result. First, resource availability is checked to place a set of SFCs in the substrate nodes (line 1). Preference lists of both SFCs and nodes are prepared (lines 2 and 3). According to \cite{Pham_2018}, nodes give higher preference to SFC requests which utilizes the maximum available resources i.e., leave out the least residual capacity unused, and SFCs give higher preference to nodes which have enough resource capacity to provide services and have higher reliability in the substrate network.

Algorithm \ref{algo:Matching algorithm} runs until all the SFCs are placed in appropriate substrate nodes. Initially, all SFCs and substrate nodes are free i.e., no SFC is assigned to any node. If there is an SFC $s_j$ which is not yet placed on any node, then $s_j$ first proposes to the highest ranked substrate node from its preference  list. Likewise, all the SFCs make proposals to their respective highest ranked substrate nodes sequentially (line 5). Each substrate node has capacity of $c_{n_i}$ and it can hold up to certain, $g$ (maximum number of SFCs that a substrate node can accommodate based on its resource capacity), number of proposals from SFCs at a time (lines 6 to 12). Nodes accept all the first $g$ number of proposals from the SFCs irrespective of their positions/ranking in the preference lists of nodes $pl(n_i)$. If a new SFC request $s_k$ comes to the substrate node after accommodating $g$ number of SFC requests and $s_k$ has higher preference than the already accommodated SFC requests, then first precompute the resource to be reclaimed by rejecting the lesser preferred SFCs before actually reject them. Hence, the already  accommodated lesser preferred SFCs are rejected only if the total (residual + reclaimed) resource is enough to accommodate the higher ranked SFC. The rejection happens sequentially with reclamation of the assigned resource of lesser preferred SFC $s^{\prime}_j$, and this process continues until $c_{n_i} + c_{s^{\prime}_j} \geq c_{s_k}$ (lines 14 to 20).  If either the SFC $s_k$ has lesser preference than the already accepted SFCs or the estimated resource by precomputation is not enough (even after reclaiming the resource of all the previously allocated SFCs), then the SFC $s_k$ is rejected (lines 21 to 24). Rejected SFC proposes to its next highest ranked substrate node from the preference list in the subsequent iterations. This procedure continues until the SFC $s_k$ is assigned to one of the preferred substrate nodes. This principle is called as deferred acceptance because initially an SFC $s_j$ can be accepted by a substrate node $n_i$ if there is resource availability and later $s_j$ can be rejected if there is a proposal from a higher ranked SFC to the same node $n_i$. The deferred acceptance based algorithm produces stable matching.

\begin{theorem}
	Many-to-one matching which employs deferred acceptance algorithm produces at least one stable matching result for general preferences such that all SFCs allowed to participate in the game are placed on the substrate network nodes.	
\end{theorem}

\begin{proof}
	We assume that substrate nodes have enough capacity to accommodate all the SFC requests \cite{Gale_Shapley}. SFCs on one side propose to the nodes based on their respective preferences.  Each node on the other side accepts all the proposals until its quota/residual capacity is over. Once the quota/residual capacity of the node is over, SFC requests are processed based on preference order of the node. An SFC request which has higher preference should be accepted, and to accommodate that the less preferred (already accepted) requests are rejected if there is not enough residual capacity on the node. A rejected SFC request is allowed to propose again to the same node till either the available resource in the node is not enough or it is less preferred than all the already accepted SFC requests. The same procedure is followed until all the SFC requests are placed in one of their preferred nodes. If an SFC would prefer to be matched to a node other than the assigned node, then according to our  algorithm, the SFC must have already proposed to  the preferred node and the preferred node must have rejected the SFC due to either lack of the resources or the SFC being less preferred than the already accepted SFCs. It means that the preferred node has another SFC that it strictly prefers, hence there cannot be a blocking pair. Therefore, our modified many-to-one matching algorithm which follows deferred acceptance procedure produces a stable matching result.    	
\end{proof}

\begin{theorem}
	Algorithm \ref{algo:Matching algorithm}, which follows deferred acceptance procedure, produces not only a stable but an optimal assignment of SFC requests onto substrate nodes. 
\end{theorem}

\begin{proof}
	In the process of deferred acceptance matching procedure, no higher ranked SFC request is rejected by any node in order to accept a lower ranked SFC request. If a higher ranked SFC is rejected due to insufficiency of resources, then the lower ranked SFCs can be accommodated for maximum utilization of resources. The node does not prefer the rejected higher ranked SFC request to the accepted lower ranked SFC request.   This procedure only rejects requests which could not be accommodated in any stable assignment. Therefore, our modified algorithm which follows deferred acceptance procedure not only yields a stable but an optimal assignment.    
\end{proof}

From theorems 2 and 3, we conclude that the modified matching algorithm produces stable and optimal assignment of SFC requests onto substrate nodes based on the preference lists. In matching theory, optimal assignment means there is no better matching/assignment than the current one. Note that optimal assignment in matching theory is different from optimal solution to the ILP (SFC placement) problem. The time complexities of  preference list preparation for SFCs  and nodes are $O(\mathcal{S} log \mathcal{S})$ and $O(\mathcal{N} log \mathcal{N})$, respectively. In the worst case, $O(\mathcal{S}^2)$ proposals are executed in each node $n \in \mathcal{N}$ in Algorithm~\ref{algo:Matching algorithm}. Hence, the time complexity of Algorithm~\ref{algo:Matching algorithm} is $O(\mathcal{N}\mathcal{S}^2)$. 

\section{Performance Evaluation}
In this section, we evaluate the performance of our proposed solutions presented in the previous sections.

\subsection{Performance Analysis of Subchaining an SFC to Enhance the Reliability}
We analyse the performance of SFC subchaining method proposed in section \ref{SFC_design} to enhance the reliability of an SFC. Simulation parameters considered in our simulation are shown in Table \ref{tab:2}. Simulation parameters are chosen such that the system is stable. For the system to be stable, traffic intensity $\rho = \lambda_s / \mu_v$ should be less than one \cite{Raj_Jain}. Reliability rate parameters $p_v$ and $p_n$ are taken from \cite{Qu1_2018} \cite{SLA} and SFC size parameter is taken from \cite{power_aware}. Simulation results are obtained using discrete-event simulator MATLAB Simulink. We compare our proposed M/M/1 and M/M/m settings with i) service chain backup setting where there is a dedicated backup for every VNF in an SFC (SCB1) and ii) service chain backup setting where a separate SFC chain is assigned as backup for a primary SFC (SCB2). We compare our results in terms of reliability, expected response time, and amount of resources required for an SFC.

\begin{table}[]
	\begin{center}
		\footnotesize
		\caption{Simulation parameters for subchaining}
		\label{tab:2}
		\begin{tabular}{|l|l|}
			\hline 
			Parameters & Values \\
			\hline
			Arrival rate, $\lambda_s$ & 100 \\
			Serving rate of VNFs, $\mu_v$ & 200 \\
			Size of SFC, $|\mathcal{V}_s|$ & 5 \\
			Reliability rate of VNFs,  $p_{v}$ & 0.9 \\
			Reliability rate of substrate nodes, $p_n$ & 0.999  \\
			\hline
		\end{tabular}
	\end{center}
\end{table}

\begin{figure*}[]
	\centering
	\begin{subfigure}{0.45\textwidth}
		\includegraphics[width=7.5cm,height=4cm]{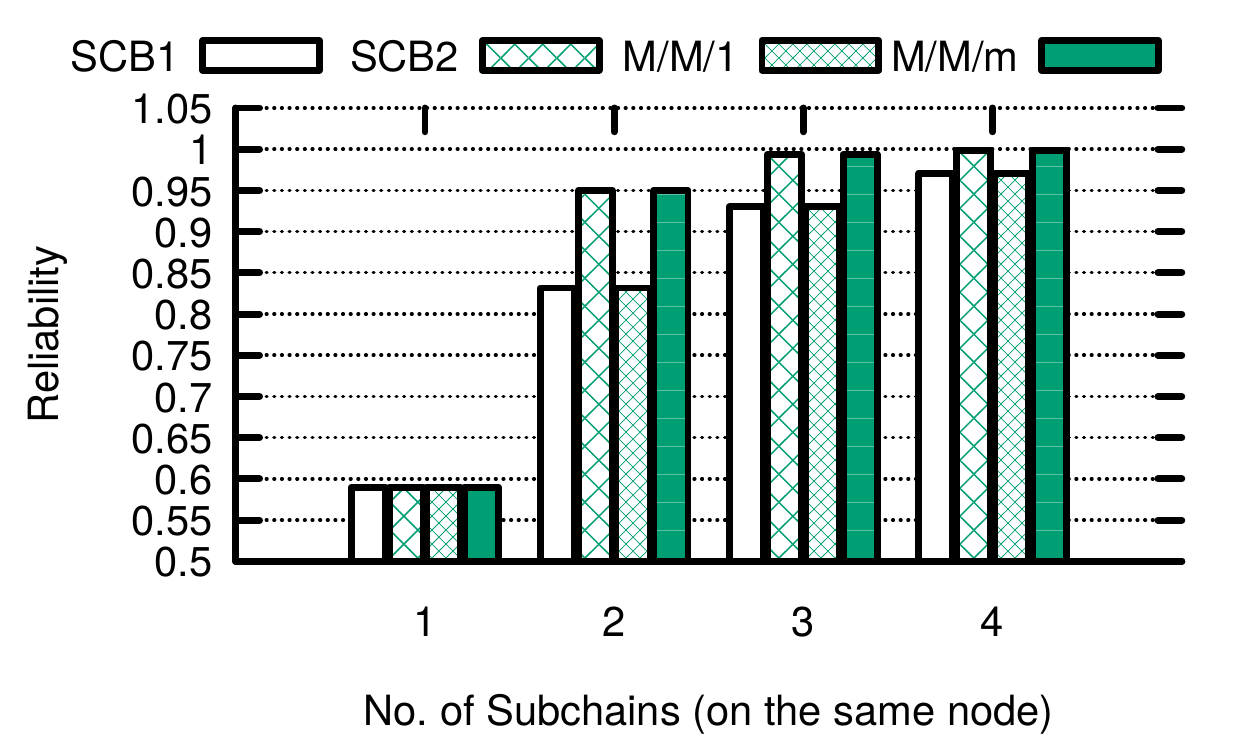}
		\caption{No. of subchains vs. Reliability}
		\label{fig:subChainsVsRel1}
	\end{subfigure} 
	\begin{subfigure}{0.45\textwidth}
		\includegraphics[width=7.5cm,height=4cm]{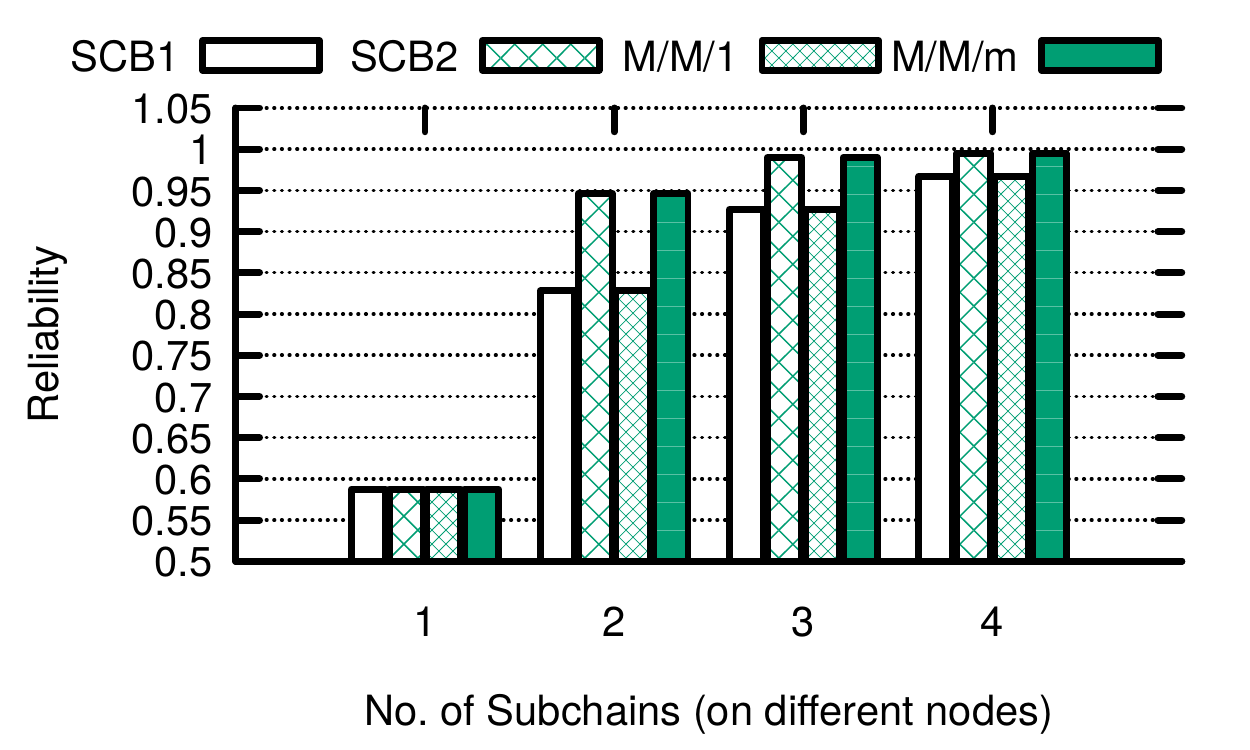}
		\caption{No. of subchains vs. Reliability}
		\label{fig:subChainsVsRel2}
	\end{subfigure}
	
	\begin{subfigure}{0.45\textwidth}
		\includegraphics[width=7.5cm,height=4cm]{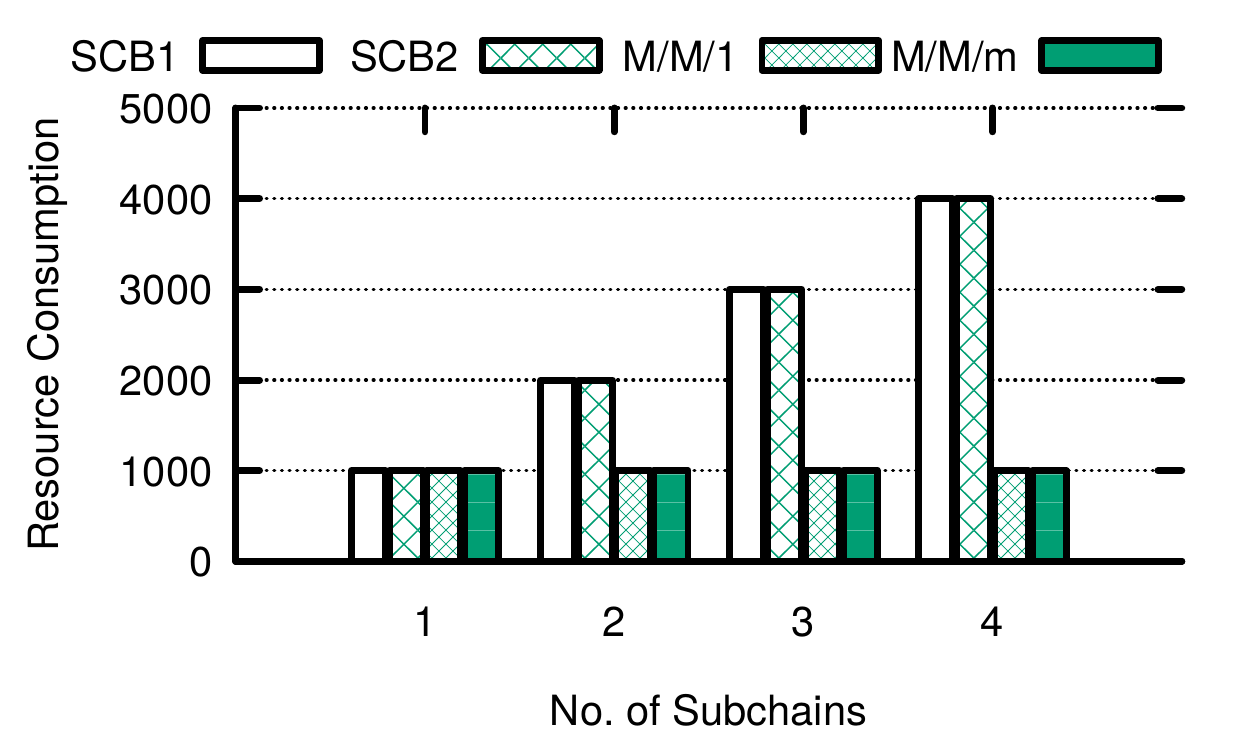}
		\caption{No. of subchains vs. Resource consumption}
		\label{fig:subChainsVsResource}
	\end{subfigure}
	\begin{subfigure}{0.45\textwidth}
		\includegraphics[width=7.5cm,height=4cm]{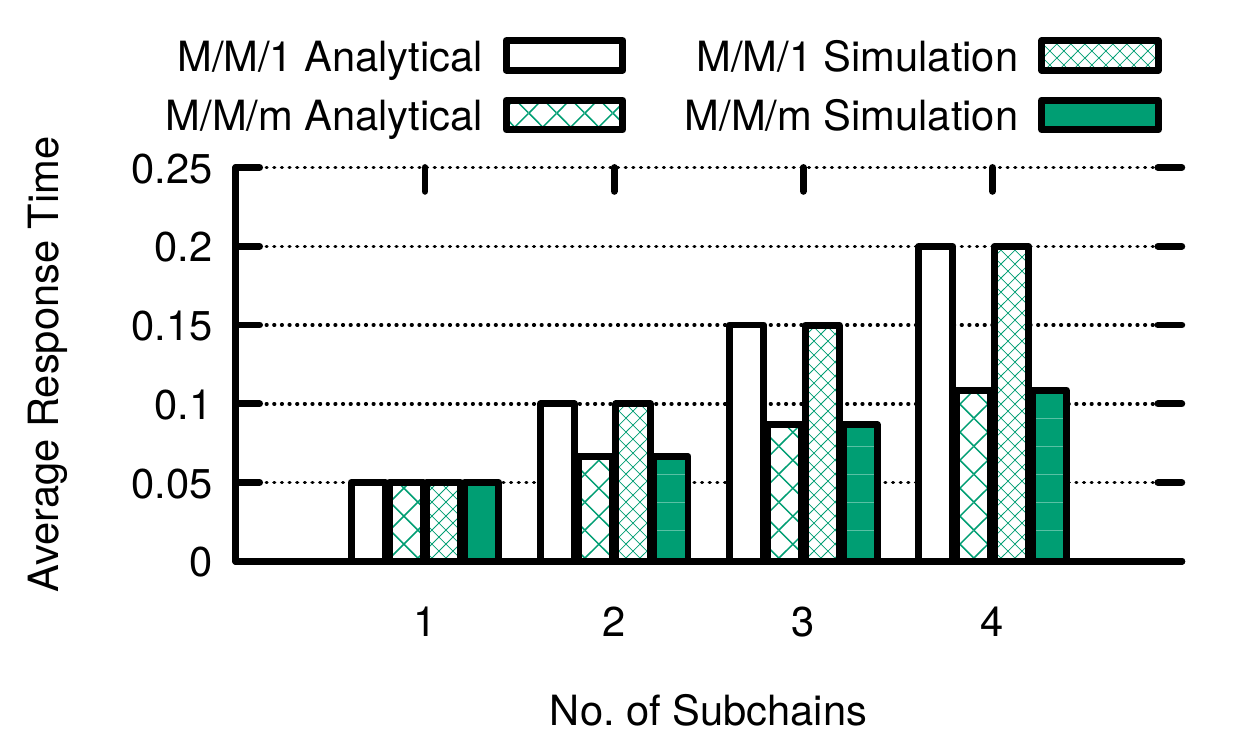}
		\caption{No. of subchains vs. Average response time (in sec)}
		\label{fig:delayVsSubchains}
	\end{subfigure}
	\caption{Subchaining analytical and simulation results.}
	
\end{figure*} 

Figures \ref{fig:subChainsVsRel1} and \ref{fig:subChainsVsRel2} show the effect of reliability on number of subchains placed on the same node and different substrate nodes, respectively. The reliability difference on placing VNFs of an SFC on the same node and on different nodes is very less. Placing subchains on different nodes provides higher robustness than placing them on the same node, however it comes with additional resource costs (a trade-off). It is clear that M/M/m setting provides higher reliability than that of M/M/1 setting. Since M/M/1 and M/M/m settings are chained based on the two actual backup settings (dedicated VNFs and separate SFC chains), they also provide the same reliability rate which is equal to SCB1 and SCB2, respectively. However, SCB1 and SCB2 methods consume more additional resources (with respect to serving rate of VNF) as we increase the number of subchains, which is shown in Figure \ref{fig:subChainsVsResource}. In the subchaining methods, we consider number of virtual cores (directly relates to $\mu_v$), assigned to all VNFs in $\mathcal{V}_s$ to process the traffic of requested service, as the resources. Since the primary SFC chain is divided into lesser capacity subchains, M/M/1 and M/M/m settings consume almost the same number of resources as a single chain primary SFC. Figure \ref{fig:delayVsSubchains} shows the effect of average response time (in seconds) on number of subchains. Simulation results obtained using MATLAB Simulink are compared with the analytical results. It is clear that M/M/m setting has less average response time than that of M/M/1 setting, and the average response time increases linearly as we increase the number of subchains in both the settings.     

We compare the performance of SFC subchaining methods with the related work \cite{Chantre1_2018}. Two redundancy models are proposed to improve the reliability of NFV-based 5G networks \cite{Chantre1_2018}: i) series-parallel backup model, in which for series of VNFs in service chain backups are assigned in parallel for each VNF (SPB) and ii) parallel-series backup model, in which an entire VNF chains are assigned as backup in parallel for series of VNFs in service chain (PSB). We consider unit cost VNFs and backup VNFs are operationally synchronized with primary VNFs and be ready to replace primary VNFs in case of failure \cite{Chantre1_2018}. Total cost to construct parallel backups for an SFC with size five are compared in Figure \ref{fig:SFC_comparison}a and it can be seen from this figure that the total SFC construction cost increases for SPB and PSB models as we increase the number of parallel backups. Since the same VNFs of an SFC is divided into lesser capacity VNFs to construct parallel backpus, the construction cost is less for M/M/1 and M/M/m based subchaining models. Mean delay of different backup models is compared in \ref{fig:SFC_comparison}b. As dedicated individual backups are assigned for SPB and PSB models, and the processing capacity of VNFs is same for both primary and backup VNFs, the mean delay is same irrespective of the number of parallel backups. Since the processing capacity is equally shared with backup VNFs, mean delay for M/M/1 and M/M/m models increases as we increase the number of parallel backups for an SFC.        

\begin{figure}[h!]
	\centering
	\begin{subfigure}{0.45\textwidth}
		\includegraphics[width=7.5cm,height=4cm]{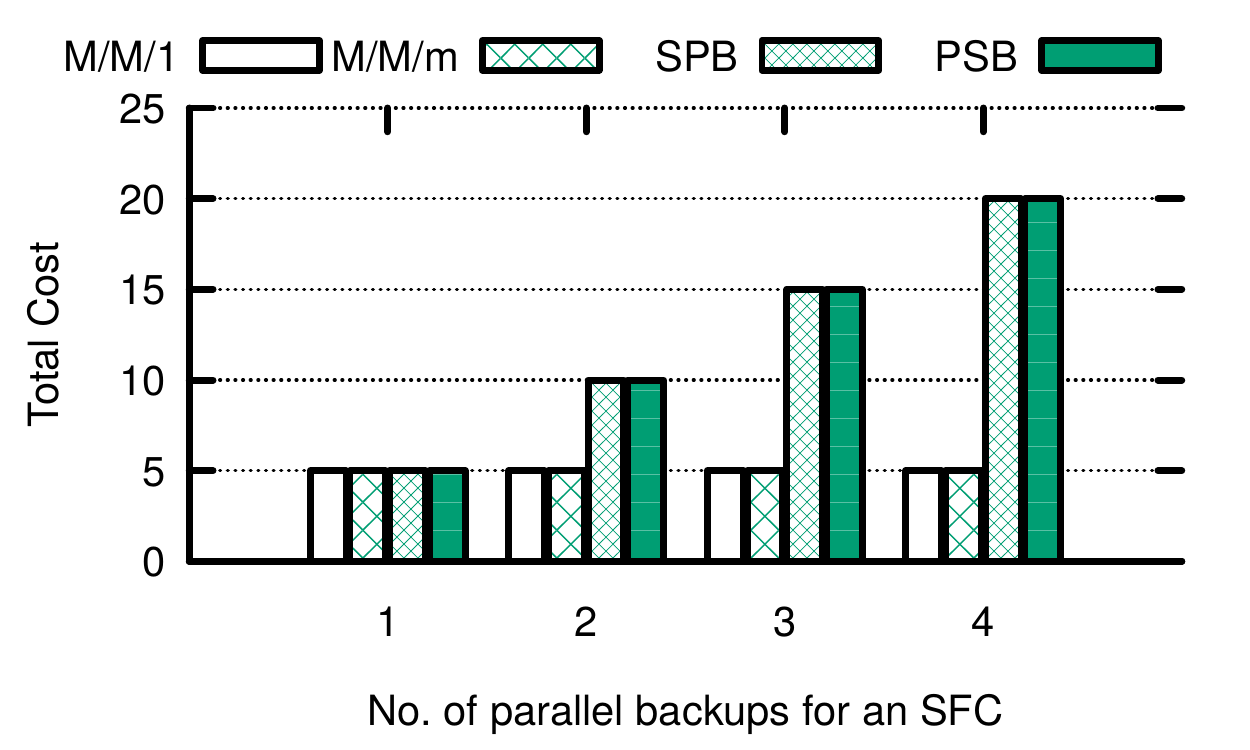}
		\caption{No. of backups vs. Total cost}
		\label{fig:subChainsVsCost}
	\end{subfigure} 
	\begin{subfigure}{0.45\textwidth}
		\includegraphics[width=7.5cm,height=4cm]{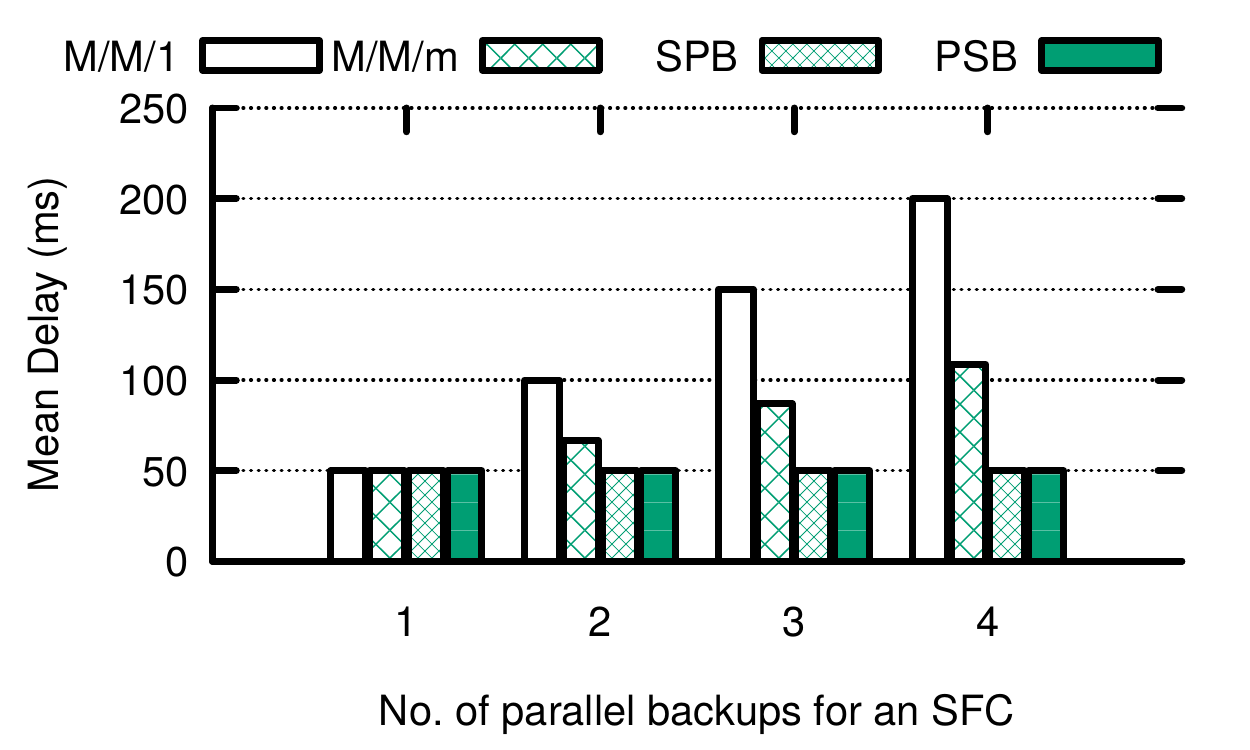}
		\caption{No. of backups vs. Mean delay}
		\label{fig:subChainsVsDelay}
	\end{subfigure}
	\caption{Performance analysis of SFC subchaining.}
	\label{fig:SFC_comparison}
\end{figure}

\begin{table*}[t]
	\begin{center}
		\footnotesize
		\caption{SFC requirements \cite{Savi}}
		\label{tab:3}
		\begin{tabular}{|l|l|c|c|c|c|}
			\hline 
			Service Types & Ordered VNFs of SFCs & Bandwidth & Delay ($\Psi_s$) & Reliability ($\Delta_s$) & Traffic\\
			\hline
			1. Web service & NAT-FW-TM-WOC-IDPS & 100 kbps & 500 ms & 0.90 & 18.2$\%$ \\
			2. Voice over IP & NAT-FW-TM-FW-NAT & 64 kbps & 100 ms & 0.999 & 11.8$\%$ \\
			3. Video streaming & NAT-FW-TM-VOC-IDPS & 4 Mbps & 100 ms & 0.99 & 69.9$\%$ \\
			4. Online gaming & NAT-FW-VOC-WOC-IDPS & 50 kbps & 70 ms & 0.99 & 0.1$\%$ \\
			\hline
		\end{tabular}
	\end{center}
\end{table*}

\subsection{Performance Analysis of Reliable SFC Design}
We analyze the performance of reliable SFC design proposed in section \ref{SFC_design}. For the evaluation purpose, we consider four service types which are \cite{power_aware} \cite{Savi}: web service, Voice over IP, video streaming, and online gaming. For each service type request, ordered list of VNFs, bandwidth, delay, reliability, and traffic generation percentage requirements are given in Table \ref{tab:3}. Six different VNFs are considered to construct SFCs, which are Network Address Translator (NAT), Firewall (FW), Traffic Monitor (TM), WAN Optimization Controller, Video Optimization Controller (VOC), and Intrusion Detection and Prevention System (IDPS). Service requests are generated based on the traffic generation percentage. For instance video streaming service has traffic generation percentage of 69.9$\%$, which is the highest.  

Each service type request has different delay $\Psi_s$ and reliability $\Delta_s$ requirements as shown in Table \ref{tab:3}. For each service type request, we assign reliability requirement similar to G Suite SLA assigned for various Google applications \cite{SLA}. For given requirements of a service request, we design a service chain based on the Algorithms \ref{algo:alg1} and \ref{algo:alg2} such that  the delay and reliability requirements are met with minimal additional backup resources without violating SLAs. Algorithm \ref{algo:alg1} uses the subchaining procedure to meet the reliability of an SFC without assigning any dedicated backup resources. Table \ref{tab:4} shows the results of SFC subchaining procedure for both M/M/1 and M/M/m settings. From the results, it is clear that SFC subchaining enhances the reliability as well as increases the processing delay. It is observed that making two subchains is enough to meet the requirements (both delay and reliability) of service type 1 web service ($\Psi_s$ = 500 ms, $\Delta_s$ = 0.90) in M/M/m setting, hence no backup is required. For other service types, we can make the maximum of only three subchains because making more than three subchains violates the delay constraint. In M/M/1 setting, making three subchains meets the reliability requirement of service type 1 (web service). For other service types, it is possible to make only two subchains without violating the delay constraints. In both M/M/1 and M/M/m settings, if reliability requirement is not met after subchaining, then backups are added in an efficient manner to meet the SLAs. 

\begin{table}[]
	\begin{center}
		\footnotesize
		\caption{SFC request subchaining results based on Algorithm \ref{algo:alg1}}
		\label{tab:4}
		
		\begin{tabular}{|p{0.01cm}|p{0.82cm}|p{1.0cm}|p{1.04cm}|p{1.03cm}|p{1.0cm}|p{1.04cm}|}
			\hline
			\multirow{2}{*}{$l$} & \multicolumn{3}{c|}{M/M/1 setting} & \multicolumn{3}{c|}{M/M/m setting} \\ \cline{2-7} 
			& Delay          & Reliability  & vCPUs        & Delay          & Reliability   & vCPUs     \\ \hline
			1 & 50 ms & 0.5899 & 5$\times$4$\times$1=20 & 50 ms & 0.5899 & 5$\times$4$\times$1=20 \\ 
			2 & 100 ms & 0.8315 & 5$\times$2$\times$2=20 & 66.7 ms & 0.9500 & 5$\times$2$\times$2=20 \\
			3 & 150 ms & 0.9304 & 5$\times$2$\times$3=30 & 86.8 ms & 0.9940 & 5$\times$2$\times$3=30 \\
			4 & 200 ms & 0.9709 & 5$\times$1$\times$4=20 & 108.7 ms& 0.9985 & 5$\times$1$\times$4=20 \\
			\hline
		\end{tabular}
	\end{center} 
\end{table}

According to \cite{power_aware}, each VNF of an SFC requires 4 vCPUs to perform an operation. If we divide the SFC into multiple subchains $l$, then each VNF in the subchain requires $\lceil{\frac{4}{l}}\rceil$ vCPUs to perform the operation. Since each service request has five ordered VNFs on a chain, the total amount of resources (vCPUs) required to place the entire chain after dividing into multiple subchains $l$ is calculated as 5$\times\lceil{\frac{4}{l}}\rceil\times l$, which is shown in the Table \ref{tab:4}.     

\begin{table*}[]
	\begin{center}
		\footnotesize
		\caption{SFC request reliability guaranteeing results based on Algorithm \ref{algo:alg2}}
		\label{tab:5}
		
		\begin{tabular}{|p{1.75cm}|c|c|c|l|c|c|c|l|}
			\hline
			\multirow{2}{*}{Service types} & \multicolumn{4}{c|}{M/M/1 setting} & \multicolumn{4}{c|}{M/M/m setting}  \\ \cline{2-9} 
			& $l$ & \#backups         & Reliability  & vCPUs   & $l$ & \#backups         & Reliability   & vCPUs   \\ \hline
			Web service & 3 & 0 & 0.9300 & 5$\times$2$\times$3 + 0 = 30  & 2 & 0 & 0.9500 & 5$\times$2$\times$2 + 0 = 20 \\ 			Voice over IP & 2 & 15 & 0.9983 & 5$\times$2$\times$2 + 15$\times$2 = 50 & 3 & 5 & 0.9990 & 5$\times$2$\times$3 + 5$\times$2 = 40 \\ 
			Video stream & 2 & 9 & 0.9924 & 5$\times$2$\times$2 + 9$\times$2 = 38 & 3 & 0 & 0.9940 & 5$\times$2$\times$3 + 0 = 30 \\ 
			Online gaming & 1 & 10 & 0.9940 & 5$\times$4$\times$1 + 10$\times$4 = 60 & 2 & 5 & 0.9940 & 5$\times$2$\times$2 + 5$\times$2 = 30 \\ \hline 
		\end{tabular}
	\end{center} 
\end{table*}

\begin{table}[h!]
	\footnotesize	
	\centering
	\caption{Average running times of Algorithms 1 and 2 in seconds}
	\label{tab: running_time}
	\begin{tabular}{|c|c|c|c|c|}
		\hline
		\multirow{2}{*}{\begin{tabular}[c]{@{}c@{}}\# of service \\  requests\end{tabular}} & \multicolumn{2}{c|}{Algorithm 1} & \multicolumn{2}{c|}{Algorithm 2} \\ \cline{2-5} 
		& M/M/1           & M/M/m          & M/M/1           & M/M/m          \\ \hline
		25 & 0.28 & 0.18 & 0.62 & 0.58 \\ \hline
		50 & 0.39 & 0.31 & 0.97 & 0.89 \\ \hline
		100 & 0.62 & 0.51 & 1.59 & 1.49 \\ \hline
		200 & 0.93 & 0.98 & 2.85 & 2.72 \\ \hline
		300 & 1.32 & 1.64 & 4.14 & 3.96 \\ \hline
		400 & 1.64 & 2.43 & 5.45 & 5.27 \\ \hline
		500 & 1.87 & 3.36 & 6.72 & 6.52 \\ \hline
	\end{tabular}
\end{table}

Although subchaining enhances the reliability of SFCs, in some cases it may not meet the reliability requirement of service requests. For instance, VoIP service request reliability requirement (required is 0.999) is not met by the subchaining procedure (obtained is 0.995). It is because of delay constraint. We proposed a novel way of satisfying reliability requirements of service requests with minimal additional redundant resources using SFC subchaining and incremental backups in Algorithms \ref{algo:alg1} and \ref{algo:alg2}, respectively. From the results shown in Table \ref{tab:5}, it is observed that to meet the SLAs (particularly reliability requirement) some service requests do not require backups while other service requests require backups. For instance, web service ($\Delta_s$ = 0.9) and video streaming ($\Delta_s$ = 0.99) service types do not require backups to meet the reliability requirements in M/M/m setting. In M/M/1 setting, web service request does not require backup but it needs three subchains $l$ (one more than the M/M/m setting) while the video streaming service type requires nine redundant backups (no backups are required in M/M/m setting) to meet the expected reliability requirement. As we can see in Table \ref{tab:5}, a maximum of fifteen and five redundant backups are required for Voice over IP service type to meet the reliability requirement in M/M/1 and M/M/m settings, respectively.  Total number of resources (vCPUs) required to meet the SLAs is shown in Table \ref{tab:5}, which is calculated as resources required after subchaining plus the resources required for redundant backups. Since web service request does not require backups, the total amount of resources required to guarantee the reliability is 30 vCPUs and 20 vCPUs in M/M/1 and M/M/m settings, respectively.  The running times of Algorithms 1 and 2 are given in Table \ref{tab: running_time}. As shown in the table, both the algorithms terminated in a few seconds, i.e., in polynomial time.

\begin{table}[]
	\begin{center}
		\footnotesize
		\caption{SFC request reliability guaranteeing results based on REACH \cite{Qu1_2018}}
		\label{tab:6}
		
		\begin{tabular}{|l|c|c|l|}
			\hline
			\multirow{2}{*}{Service types} & \multicolumn{3}{c|}{REACH \cite{Qu1_2018}} \\ \cline{2-4} 
			&  No. of backups         & Reliability  & vCPUs     \\ \hline
			1. Web service & 5 & 0.9500 & 5$\times$4 + 5$\times$4 = 40  \\ 
			2. Voice over IP & 15 & 0.9990 & 5$\times$4 + 15$\times$4 = 80 \\
			3. Video streaming & 10 & 0.9940 & 5$\times$4 + 10$\times$4 = 60  \\
			4. Online gaming & 10 & 0.9940 & 5$\times$4 + 10$\times$4 = 60 \\
			\hline
		\end{tabular}
	\end{center} 
\end{table}

\begin{figure}[]
	\centering
	\begin{subfigure}{0.45\textwidth}
		\includegraphics[width=7.5cm, height=4cm]{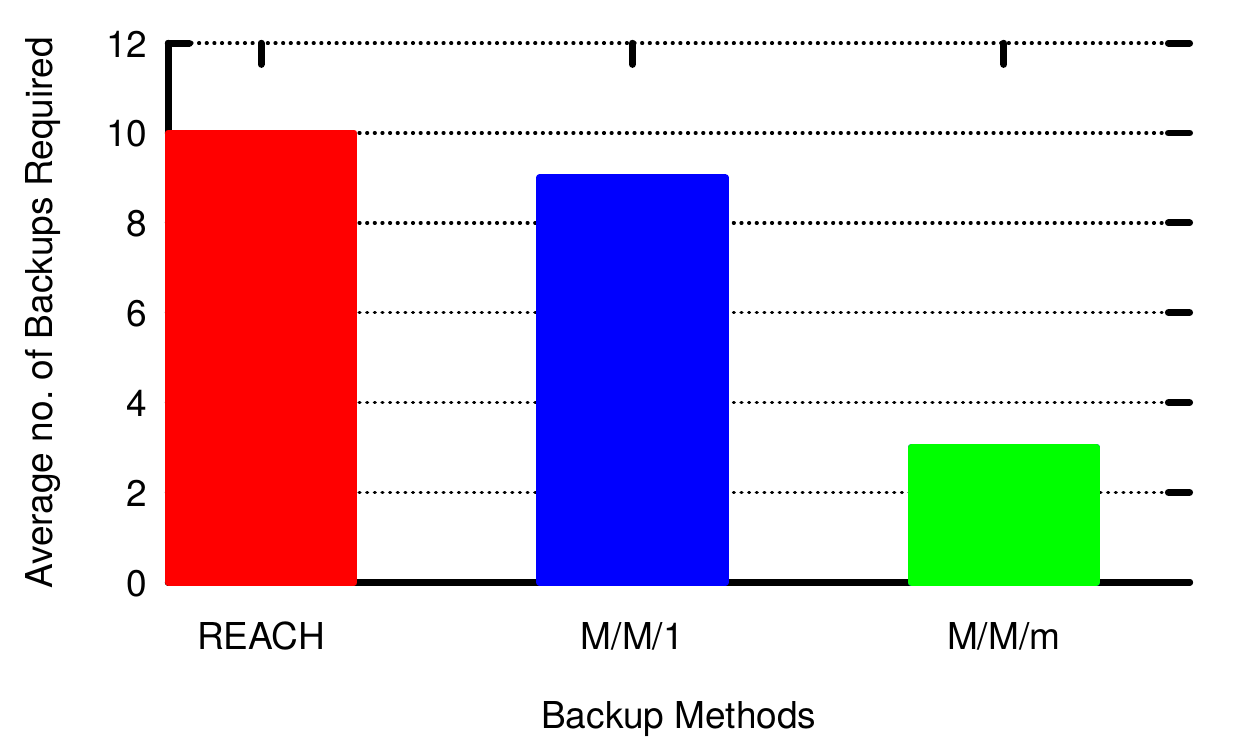}
		\caption{Average no. of backups required to satisfy the reliability requirement}
	\end{subfigure} 
	\begin{subfigure}{0.45\textwidth}
		\includegraphics [width=7.5cm, height=4cm]{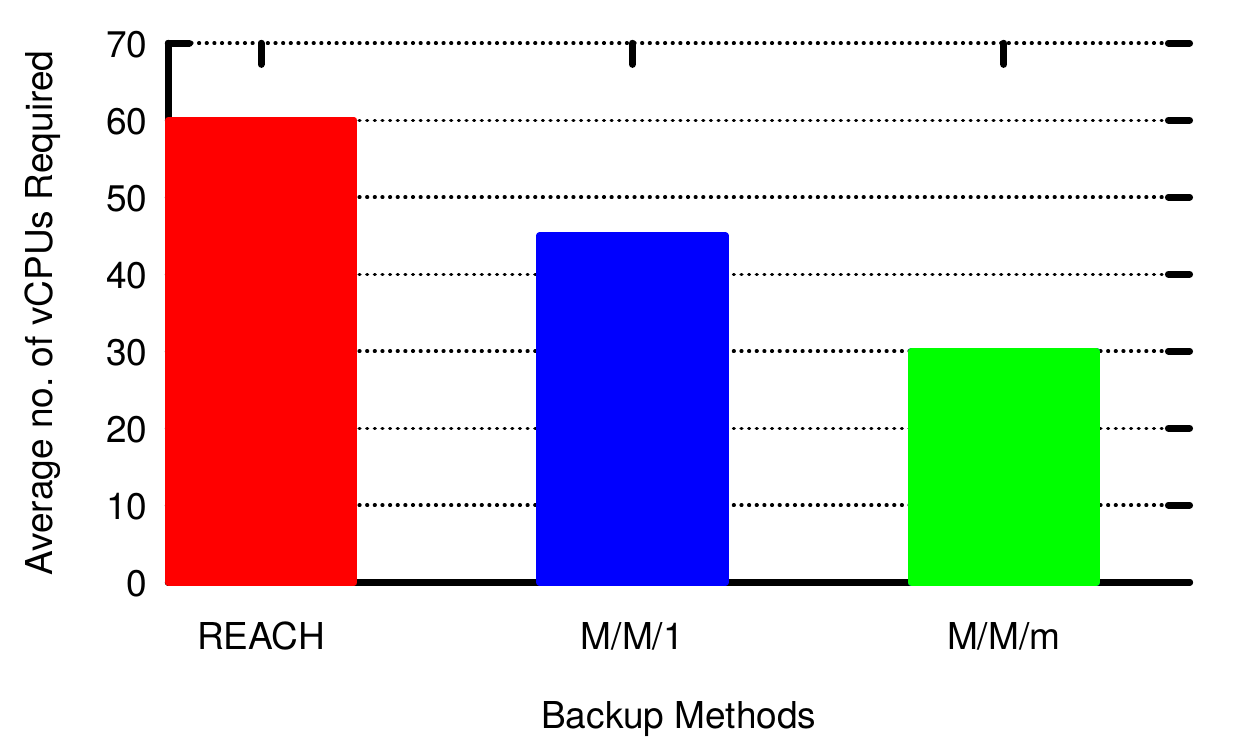}
		\caption{Average no. of vCPUs required for reliable SFC placement}
	\end{subfigure} 
	\caption{Performance comparison with REACH \cite{Qu1_2018}.}
	\label{Fig:com}
\end{figure}

For the same setup, we calculate the number of backups required to meet the reliability requirement of service requests and the number of resources required to place the primary and backup network functions based on the solution proposed in a state-of-the-art related work REACH \cite{Qu1_2018}. The results are shown in Table \ref{tab:6}. From Figure \ref{Fig:com}, it is clear that our proposed solutions perform better than REACH \cite{Qu1_2018} in terms of number of backups and resources required for reliable placement. Since we divide the SFC into multiple subchains of lesser capacity, the resource required for backups in our design is much lesser than the resource requirement of original undivided network functions. Therefore, compared to REACH \cite{Qu1_2018} our proposed solution consumes 25$\%$ to 50$\%$ lesser redundant resources for M/M/1 and M/M/m settings, respectively.  

\subsection{Performance Analysis of Placement of Reliable SFCs}
We analyse the performance of reliable SFC design placement proposed in section \ref{placement}. Since M/M/m setting performs better than M/M/1 setting as shown in Tables \ref{tab:4} and \ref{tab:5}, we consider only the placement of subchains of M/M/m setting along with backups based on the design results of Algorithms \ref{algo:alg1} and \ref{algo:alg2}. If a VNF is shared to create multiple service chains to serve multiple requests, then a failure of one common VNF may bring down all the chains. To avoid disruption of multiple services due to a single network element failure and to enhance the reliability, an individual SFC chain is created for each service request. Note that we have assumed that the primary VNFs and their corresponding backup VNFs of an SFC chain are placed on a single substrate node to easily activate the backups and synchronize the primary VNFs operations with backup VNFs, and to reduce energy consumption, bandwidth consumption, and inter VNF communication delays.  

We consider that each substrate node has the resource capacity of 28 CPU cores and it is enabled by hyper-threading \cite{Server}. Therefore, 56 vCPUs are available at each substrate node to host the virtual nodes. We assume that the underlying substrate network has 400 nodes to accommodate multiple SFCs \cite{Knight}. We consider four types of service requests as shown in Table \ref{tab:3}. Resources required for each service request to meet the SLAs are shown in Table \ref{tab:5}. For an M/M/m setting, the minimum and maximum vCPUs required for different service types are 20 and 40, respectively. To have various vCPU requirement size, we randomly generate vCPU requirement between 20 and 40 for each service request. 

\begin{figure}[]
	\centering
	\begin{subfigure}{0.45\textwidth}
		\includegraphics[width=7.5cm,height=4cm]{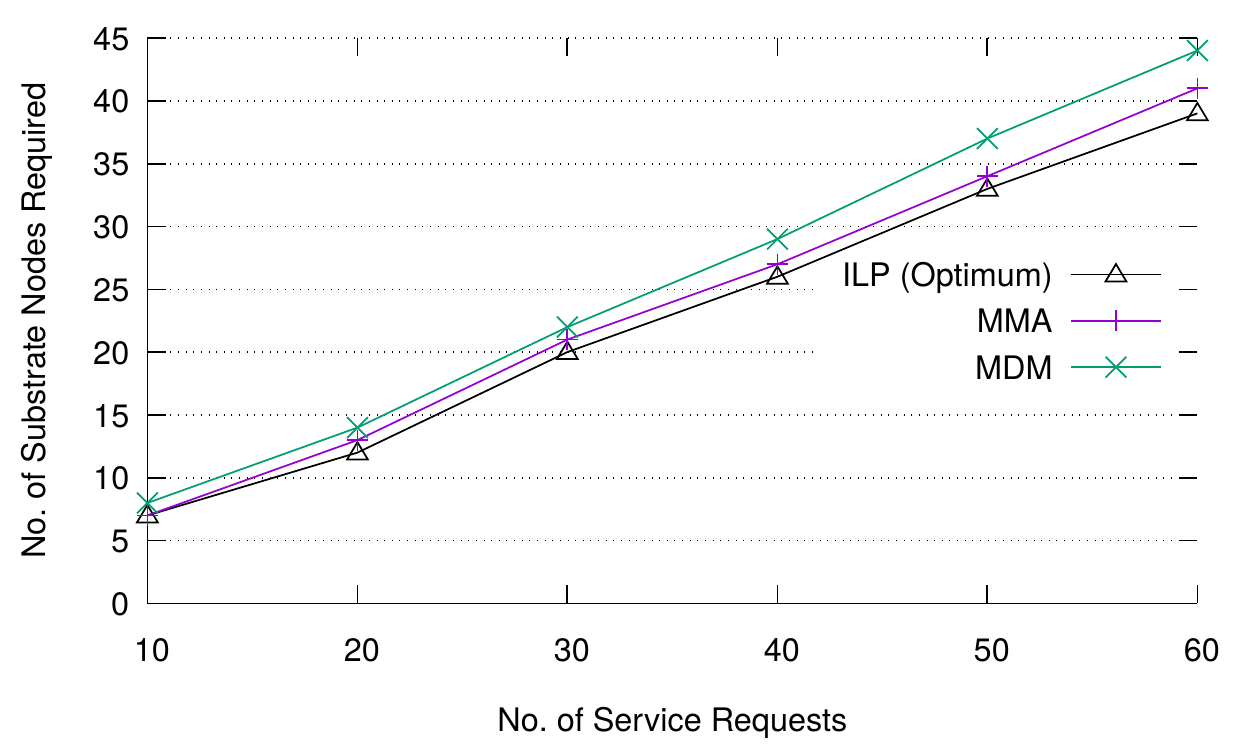}
		\caption{No. of substrate nodes required for smaller service requests}
	\end{subfigure} 
	\begin{subfigure}{0.45\textwidth}
		\includegraphics[width=7.5cm,height=4cm]{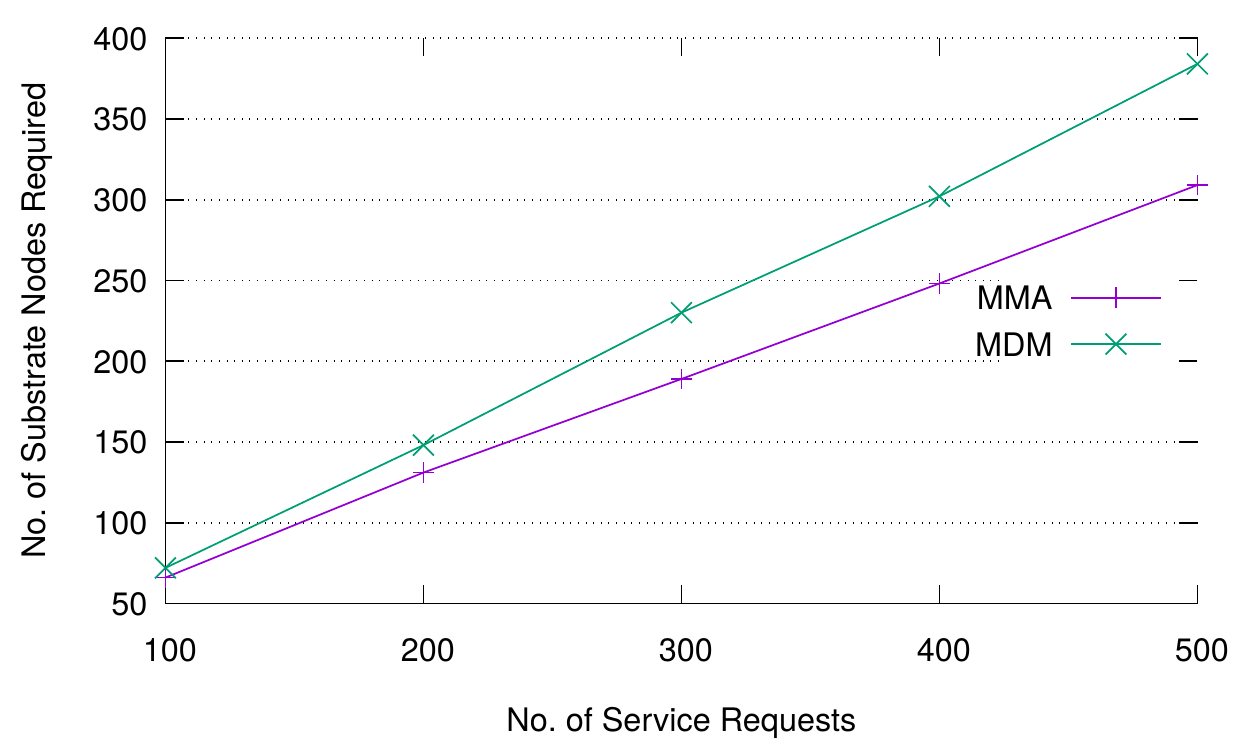}
		\caption{No. of substrate nodes required for larger service requests}
	\end{subfigure}
	\caption{Performance analysis of reliable SFC placement algorithms.}
	\label{fig:Objective_comparison}
\end{figure}

We use JuMP \cite{JuMP} for modeling ILP problem and Gurobi as solver to solve the ILP optimization problem. We compare the performance of our modified matching algorithm (MMA) with ILP (optimum), and multi-dimension matching algorithm (MDM) \cite{Pham_2018}. The comparison is done in terms of objective value (i.e., number of substrate nodes required for provisioning services) and running time. As shown in Figure \ref{fig:Objective_comparison}a, it is clear that MMA provides near-optimal solution which is closer to the optimal solution of ILP. As MMA allows to propose the SFCs as long as the resource is available on the nodes and utilizes the available resources of activated nodes efficiently, it performs better than MDM. Figure \ref{fig:Objective_comparison}b shows the results for large service requests, in which MMA is compared with MDM. The results show that MMA performs better than MDM and the performance gap increases with increasing number of service requests. We compare the running time of ILP and matching-based placement algorithms in Table \ref{tab:7}. It can be seen that as we increase the number of service requests in the network, the running time increases exponentially for the ILP model. On the other hand, matching algorithm based solutions take only a few seconds, which is significantly much lesser than the ILP. It can be seen that ILP provides optimal solution in reasonable time for small input instances. However, ILP takes longer time to converge for large input instances and hence not viable for practical deployment. Owing to high time complexity of the ILP, we designed many-to-one matching algorithm based MMA method to provide near-optimal solution in polynomial time. Since the SFCs are allowed to propose to the rejected nodes again as long as the resource is available, number of rejections in MMA is higher than in MDM. Hence, MMA takes more time than MDM as we increase the number of service requests. Although our MMA solution takes slightly more time than MDM, MMA requires less number of substrate nodes to place SFCs compared to MDM. As it can be seen in Figure \ref{fig:Objective_comparison}b, the gap increases as we increase the number of service requests. In terms of percentage, our MMA algorithm requires 8$\%$ to 24$\%$ lesser physical resources than MDM for placement of reliable SFCs.   

\begin{table}[]
	\begin{center}
		\footnotesize
		\caption{Average running time comparison of ILP and matching algorithms (in seconds)}
		\label{tab:7}
		\begin{tabular}{|p{1cm}|c|c|c|c|c|c|}
			\hline 
			\#Service requests & 10 & 20 & 30 & 40 & 50 & 60 \\ \hline  
			ILP & 0.739  & 19.799 & 75.498 & 189.762 & 735.321 & 2143.962 \\ \hline 
			MDM & 0.03 & 0.046 & 0.091 & 0.151 & 0.242 & 0.334 \\ \hline   
			MMA & 0.024 & 0.049 & 0.109 & 0.17 & 0.28 & 0.394 \\ \hline
		\end{tabular}
	\end{center}	
\end{table}

\section{Conclusion}
In this work, we focused on reliability assured, delay-guaranteed, and resource efficient SFC placement problem. We solved this problem in two phases. In the first phase, we proposed a novel method for reliable SFC design with the objective of minimizing the additional redundant resources while meeting the SLAs, and in the second phase we formulated the reliable SFC placement problem using ILP to minimize the physical resources and proposed a matching algorithm based solution to overcome the computational complexity of ILP in large input instances. Through extensive simulations we showed that our proposed solution outperforms the state-of-the-art solutions. Compared to the existing works, our reliable SFC design technique requires very less number of additional redundant resources to assure the required reliability while meeting SLAs, and our reliable SFC placement technique is more efficient and consumes minimal physical resources for provisioning the reliable communication services. We plan to extend this work by relaxing the assumption that the links between physical nodes and virtual nodes are completely reliable. Also, efficiently placing VNFs of an SFC in various data center locations under different administrative domains with different costs is an interesting and challenging problem that we plan to address in our future work. 

\section*{Acknowledgment}
The authors would like to thank the anonymous reviewers whose comments and suggestions have greatly improved this manuscript. This research work was supported by the Department of Science and Technology (DST), New Delhi, India. 

\bibliographystyle{IEEEtran}
\bibliography{ref1}

\end{document}